\documentclass[12pt]{article}

\usepackage[top=1.25in, bottom=1.25in, left=1.25in, right=1.25in]{geometry}

\usepackage{graphicx}

\usepackage[colorlinks,citecolor=black,urlcolor=black]{hyperref}

\usepackage{enumerate}
\usepackage[normalem]{ulem}
\usepackage{amsmath,amssymb,stmaryrd}
\usepackage{setspace}
\usepackage[amsthm]{newtx}
\usepackage{slantsc}\usepackage[cmtip,all]{xy}
\usepackage{accents}
\usepackage[cmyk]{xcolor}
\usepackage{titlesec}

\usepackage{epsfig}
\usepackage{sgame}
\usepackage{subcaption}

\usepackage{pgfplots}

\usetikzlibrary{calc,arrows,automata,shapes.misc,shapes.arrows,chains,matrix,positioning,scopes,decorations.pathmorphing,shadows}
\usepackage{comment}

\usepackage{natbib}
\setcitestyle{authoryear,open={(},close={)}}

\newtheorem{notation}{Notation}
\newtheorem{lemma}{Lemma}

\newtheorem{remark}{Remark}

\newtheorem{corollary}{Corollary}

\newtheorem{proposition}{Proposition}
\newtheorem{observation}{Observation}

\theoremstyle{definition}
\newtheorem{example}{Example}
\newtheorem{definition}{Definition}

\newcommand{\argmax}{\mathrm{argmax}}
\newcommand{\argmin}{\mathrm{argmin}}
\DeclareMathOperator{\Tr}{Tr}

\newcommand{\brac}[1]{\left[#1\right]}
\newcommand{\paren}[1]{\left(#1\right)}
\newcommand{\curlyb}[1]{\left\{#1\right\}}
\newcommand{\real}{\mathbb{R}}
\newcommand{\var}{\mathbb{V}}

\newcommand{\bbE}{\mathbb{E}}
\newcommand{\expec}[1]{\bbE\brac{#1}}

\newcommand{\mat}[2]{\real^{#1\times#2}}
\newcommand{\partialop}[1]{\tfrac{\partial}{\partial#1}}

\newcommand{\scalar}{\beta}
\newcommand{\func}{\varphi}

\newcommand{\amat}{D}
\newcommand{\co}{\mathrm{co}} 

\newcommand{\bign}{n}
\newcommand{\litn}{k}
\newcommand{\litnt}{{\tilde{\litn}}}
\newcommand{\optlitnb}{\litn^*}
\newcommand{\litnb}{\optlitnb}
\newcommand{\optlitnw}{\litn_*}
\newcommand{\litnw}{\optlitnw}
\newcommand{\hoptlitnb}{{\litn}^*_1}\newcommand{\loptlitnb}{{\litn}^*_2}\newcommand{\hlitn}{\bar{\litn}}
\newcommand{\llitn}{\underline{\litn}}

\newcommand{\reallit}{\real^{\litn}}
\newcommand{\somen}{m}

\newcommand{\state}{\boldsymbol\omega}
\newcommand{\stateb}{\tilde{\state}}
\newcommand{\bias}{b}
\newcommand{\stater}{\omega}

\newcommand{\defstate}{\hat\stater}
\newcommand{\staterow}{\state^{R}}
\newcommand{\statenull}{\state^{N}}
\newcommand{\mvar}{V}
\newcommand{\mvarb}{\tilde\mvar}
\newcommand{\rmvar}{\sqrt{\mvar}}
\newcommand{\act}{\mathbf{a}}
\newcommand{\actr}{a}
\newcommand{\defact}{\hat\actr}

\newcommand{\vect}{x}
\newcommand{\vectb}{y}
\newcommand{\mess}{\mathbf{m}}
\newcommand{\messr}{m}

\newcommand{\slo}{S}
\newcommand{\rlo}{R}

\newcommand{\slob}{\tilde\slo}
\newcommand{\rlob}{\tilde\rlo}
\newcommand{\Loss}{\mathbb{L}}
\newcommand{\hslo}{\slo_1}\newcommand{\lslo}{\slo_2}

\newcommand{\sstrat}{C}
\newcommand{\sstratb}{\tilde\sstrat}

\newcommand{\rstrat}{D}
\newcommand{\rstratb}{\tilde\rstrat}
\newcommand{\rstrats}{\mat\bign\litn}
\newcommand{\Rstrat}{\mathcal{D}}
\newcommand{\rstratcol}{d}
\newcommand{\strats}{A}
\newcommand{\stratsb}{\tilde\strats}
\newcommand{\stratsset}{\mathcal{A}}
\newcommand{\epset}{\mathcal{E}}

\newcommand{\cost}{\gamma}
\newcommand{\hcost}{\cost_H}
\newcommand{\lcost}{\cost_L}
\newcommand{\mcost}{\cost_M}
\newcommand{\hhcost}{\cost_{HH}}
\newcommand{\llcost}{\cost_{LL}}

\newcommand{\someeigen}{\mu}
\newcommand{\seigen}{\sigma}
\newcommand{\reigen}{\rho}

\newcommand{\seig}{\seigen}
\newcommand{\reig}{\reigen}

\newcommand{\ur}{\mathbf{u}_R}
\newcommand{\us}{\mathbf{u}_S}
\newcommand{\goodur}{\bar{u}^\litn_R}
\newcommand{\badur}{\underline{u}^\litn_R}

\newcommand{\rpay}{r}
\newcommand{\spay}{s}
\newcommand{\lowur}{\underline\rpay}
\newcommand{\highur}{\bar\rpay}
\newcommand{\medur}{\tilde\rpay}
\newcommand{\lowurb}{\rpay_L}
\newcommand{\highurb}{\rpay_H}
\newcommand{\medurb}{\rpay_M}
\newcommand{\issur}{v}

\newcommand{\aissur}{\hat\issur}
\newcommand{\hissur}{\issur^1}\newcommand{\lissur}{\issur^2}

\newcommand{\View}{P}
\newcommand{\Views}{\mathcal{P}}
\newcommand{\view}{p}
\newcommand{\Viewb}{Q}
\newcommand{\Viewt}{\tilde\View}
\newcommand{\Viewbt}{\tilde\Viewb}
\newcommand{\viewb}{q}
\newcommand{\viewt}{\tilde\view}
\newcommand{\Viewh}{\hat\View}
\newcommand{\hView}{\View_1}\newcommand{\lView}{\View_2}\newcommand{\wViewb}{W}

\newcommand{\basis}{b}

\newcommand{\gain}{G}
\newcommand{\somevec}{x}
\newcommand{\somenhd}{\mathcal{N}}
\newcommand{\somevect}{\tilde\somevec}
\newcommand{\somevech}{\hat\somevec}
\newcommand{\somevecb}{y}
\newcommand{\vectspace}{\mathcal X}
\newcommand{\vectset}{X}
\newcommand{\indset}{J}
\newcommand{\indsetb}{K}
\newcommand{\indsetset}{\mathfrak{J}}
\newcommand{\rootslo}{\sqrt{\slo}}

\newcommand{\excor}{\rho}

\newcommand{\somemat}{M}
\newcommand{\somematt}{\tilde\somemat}
\newcommand{\somemath}{\hat\somemat}
\newcommand{\matset}{\mathcal{M}}
\newcommand{\hsomemat}{\somemat_1}\newcommand{\lsomemat}{\somemat_2}\newcommand{\perm}{\Pi}
\newcommand{\permh}{\hat\perm}
\newcommand{\permt}{\tilde\perm}

\newcommand{\somedmat}{\Lambda}

\newcommand{\someorth}{G}
\newcommand{\somebi}{B}
\newcommand{\somesk}{Z}

\newcommand{\ran}{\mathrm{ran } }
\newcommand{\rank}{\mathrm{rank}}
\newcommand{\matnn}{\mat\bign\bign}
\newcommand{\matnk}{\mat\bign\litn}
\newcommand{\matkn}{\mat\litn\bign}

\newcommand{\ip}[2]{\langle#1,#2\rangle}
\newcommand{\inv}[1]{{#1}^{-1}}
\newcommand{\cormat}[2]{\mathfrak{R}(#1,#2)}
\newcommand{\simmat}[2]{\mathfrak{R}^2(#1,#2)}
\newcommand{\codet}[2]{\mathfrak{R}^2(#1,#2)}
\newcommand{\codett}{\mathfrak{R}^2}
\newcommand{\elem}{e}
\newcommand{\normproj}{E}
\newcommand{\wtset}{\mathcal{K}}
\newcommand{\wt}{\kappa}
\newcommand{\wtb}{\lambda}

\newcommand{\twomat}[4]{\begin{bmatrix} 
	#1 & #2 \\
	#3 & #4 \\
	\end{bmatrix}}
\newcommand{\twovec}[2]{\begin{pmatrix} 
	#1\\
	#2\\
	\end{pmatrix}}
\newcommand{\threevec}[3]{\begin{pmatrix} 
		#1\\
		#2\\
		#3\\
\end{pmatrix}}
\newcommand{\threemat}[9]{\begin{bmatrix} 
		#1 & #2 & #3 \\
		#4 & #5 & #6 \\
		#7 & #8 & #9 \\
\end{bmatrix}}

\newcommand{\acomment}[1]{}

\begin{document}

\title{{\bf Communicating about Endogenous Issues}\footnote{We would like to thank Ricardo Alonso, Jeff Ely, Ben Golub, Daniel Gottlieb, Faruk Gul, Yoram Halevy, Alex Jakobsen, Emir Kamenica, Navin Kartik, Marina Halac, Michael Lipnowski, Alessandro Pavan, Wolfgang Pesendorfer, Phil Reny, Joel Sobel, and various  seminar and conference audiences for extremely helpful feedback.}}
\author{
\begin{minipage}{0.3\textwidth}\centering  
Elliot Lipnowski\footnote{\texttt{website: 
https://elliotlipnowski.com, email: elliot.lipnowski@yale.edu}} \\ \centering \it \small Yale University
\end{minipage}                  
\begin{minipage}{0.3\textwidth}\centering 
Doron Ravid\footnote{\texttt{website: https://doronravid.com, email: doronr@umich.edu}}  \\ \centering \it \small University of Michigan 
\end{minipage} 
}

\date{August 2026}

\maketitle

\begin{abstract}
    Limited attention forces organizations to decide not only how much to discuss, but also which issues merit discussion. We study strategic communication about a multidimensional decision when a receiver can respond only along a few endogenously chosen issues. Players agree on the ideal action but prioritize different errors. In equilibrium, communicated and omitted issues must be statistically unrelated and separable according to the sender's preferences. Thus, the sender’s priorities determine the agenda; the receiver’s do not. Synchronized priorities raise the receiver’s best equilibrium payoff but lower his worst, so he may prefer a less synchronized sender.
\end{abstract}

\onehalfspacing

\section{Introduction}

Organizations cannot discuss every aspect of every decision. A CEO cannot learn every local condition known to a division manager; a school board cannot deliberate over every detail known to a union representative; and a policymaker cannot absorb every implication of an expert's analysis. Limited attention therefore forces people to communicate using summaries. But a summary does more than leave information out. It also groups the underlying facts into a small number of issues.

This observation has deep roots in organization theory. \citet{march1958organizations} use the term \emph{uncertainty absorption} for communication in which an inference drawn from evidence is transmitted in place of the evidence itself. An information-processing view emphasizes the need to match communication capacity to organizational uncertainty and interdependence \citep{tushman1978information}. Research on strategic issue diagnosis studies how data and stimuli are translated into focused issues that organize attention \citep{dutton1983strategic}, while the interpretation-systems and attention-based views emphasize how organizations give information meaning and direct decision makers toward particular ``issues and answers'' \citep{daft1984interpretation,ocasio1997attention}. 

Together, these theories suggest that organizational summaries solve two problems at once: they economize on attention by limiting how much is discussed, and they organize attention by determining what the discussion is about. This distinction is especially pertinent for multidimensional decision problems, which rarely come with a unique list of issues. The same facts can be summarized using different aggregates, comparisons, or categories, each of which makes some aspects of the decision transparent and others obscure. This paper studies which divisions of the decision---that is, which issues---will be transmitted in an organization when communication is strategic. 

For concreteness, consider a firm owner deciding how much money to allocate to the development of each of two products. The firm's manager knows this ideal pair, but the owner has the time and energy to understand only one figure. That figure could for instance describe the ideal budget for the first product, the ideal total budget for developing the two products, or the ideal difference between the two product budgets. Each report uses the same communication capacity. Nevertheless, each frames the owner's decision in a different way. 

What figure should we expect the manager to communicate to the owner? Our analysis shows the answer depends on the manager's perspective. For a heuristic illustration, suppose the manager cares only about assigning the division the right total budget, while the owner cares mainly about the first product. Intuitively, if the owner were to use the manager's report to set the budget of the first product alone, the manager would inflate that report when the second product needs more money. In effect, she uses the first product's budget to partially correct the total. By contrast, if the owner uses the report only to set the division's total budget, the manager should have no reason to slant it: errors in either direction move the total away from its ideal level. Thus, it is the manager's perspective on the decision that determines which issues can be credibly communicated.

Our paper formalizes this logic within a model of strategic communication with bounded rationality. A receiver consults an informed sender about a multidimensional action. Both players would like the action to match an unknown state perfectly. Although they perfectly agree on the ideal action in each state, they have different priorities, and hence evaluate imperfectly tailored actions differently. For example, both the owner and the manager agree on the ideal budget for each product, but one cares more about getting a particular product's budget right while the other cares more about setting the correct the total. Finally, the receiver is boundedly rational: he must choose a low-dimensional interpretation of the sender's message. 

Formally, the state and action are vectors $\state,\act\in\real^\bign$. We normalize the state's mean to zero and assume it has finite variance. The sender and receiver have quadratic losses,
\[
 (\act-\state)^\top \slo(\act-\state)
 \qquad\text{and}\qquad
 (\act-\state)^\top \rlo(\act-\state),
\]
where the positive-definite matrices $\slo$ and $\rlo$ describe their respective priorities. The sender observes $\state$ and sends a message $\mess\in\real^\litn$, where $\litn<\bign$. The receiver interprets messages through a full-rank linear decision rule $\rstrat$, taking action $\act=\rstrat\mess$. The $\litn$ columns of $\rstrat$ are the issues the receiver understands the sender to be discussing. So $\litn$ measures the capacity of the communication channel, while the range of $\rstrat$ describes its substantive orientation. Thus, the receiver can choose how to interpret the sender, but he must settle on a $\litn$-dimensional way of doing so.

We show equilibrium communication is shaped by the sender's \emph{perspective}. A perspective is a basis that decomposes the decision into statistically unrelated issues. A sender perspective is one for which she finds issues also unrelated in a payoff sense: in these coordinates, her loss is a weighted sum of squared errors, with one weight for each issue and no cross-issue spillovers. Every selection of $\litn$ issues from a sender perspective can be communicated in equilibrium. In such an equilibrium, the sender reports the realized values of those issues, whereas the receiver takes her report on them at face value and takes the prior-optimal action on the remaining issues. Conversely, every equilibrium outcome can be represented in this way. Hence, equilibrium communication is shaped by the sender's priorities and the underlying state distribution, but not by the receiver's priorities.

This result follows from the players' optimality conditions. If the receiver's action space mixes dimensions to which the sender assigns different weights, the sender wants to tilt her report toward the dimensions she considers more important. Truthful communication is sustainable only when the communicated and uncommunicated components are separable from the sender's point of view. Meanwhile, the receiver faces  a signal-extraction problem with the sender's equilibrium report as the signal. The receiver's incentives therefore require the communicated and residual components to be statistically unrelated, while the sender's incentives determine which decompositions are credible.

Let us return to the owner and manager, and consider the implications of our result when the two products' ideal budgets are independent and identically distributed. Suppose the manager evaluates errors in the total budget more heavily than errors in the difference between the product budgets, and cares about these two issues separably. Her perspective then consists of exactly these two issues: the total budget and the difference. If the owner can contemplate only one issue, some equilibrium exists in which the manager communicates the ideal total, and another exists in which she communicates the ideal difference. In contrast, no equilibrium has the manager simply reporting the first product's ideal budget, even though the owner cares more about the first product than the second.

Would the owner be better off with a manager whose perspective is more similar to his own? We show the answer depends on whether the manager expects his favorite or least favorite equilibrium. For a heuristic demonstration, suppose the owner places weights $98$ and $2$ on matching the ideal budget of the first and second product, respectively. Because either issue in the manager's perspective mixes the two products symmetrically, the owner obtains a payoff of $50$ in either equilibrium. But what if the owner instead consults a manager who shares his perspective? Such a manager could communicate either product's ideal budget, giving the owner a payoff of $98$ in one equilibrium and $2$ in the other. 

We show the above paragraph demonstrates a more general phenomena: a sender who is better for the receiver than another under best-equilibrium selection tends to be worse under adversarial selection. The broad intuition is that improving the receiver's best-case payoff requires the sender's perspective to better separate the receiver's high- and low-priority concerns. Such a separation stretches the set of equilibrium payoffs: it makes especially useful conversations possible while also making it possible to focus on minutiae. 

We formalize this idea by defining a synchronization order between sender perspectives, and establishing for every given receiver perspective a three way equivalence between: (1) one sender being more synchronized with that perspective than another sender; (2) one sender giving a higher best-equilibrium payoff than another for all receivers with that perspective; and (3) one sender giving a lower worst-equilibrium payoff than another to all receivers with that perspective. Thus, while seeing eye-to-eye is better when the players talk about high-priority issues, having different worldviews provides insurance against coordinating on low-priority topics.

Finally, we allow the receiver to choose how much attention to devote to the sender. Under favorable equilibrium selection, the receiver learns about his favorite $\litn$ sender issues. Attention therefore has decreasing marginal returns, and an interior capacity can be optimal. Under adversarial selection, he expects to hear about his least valuable sender issues. Thus, attention instead has increasing marginal returns, so the receiver either listens fully or not at all. Moreover, the average value of an issue is invariant to the sender's perspective in this case. Consequently, the receiver's attention choice under adversarial selection does not depend on whom he consults.

\paragraph{Related literature.}
We build on the cheap-talk framework of \citet{Crawford1982} and \citet{Green2007two}. In multidimensional cheap talk, equilibrium often determines the directions along which information is conveyed. Multiple experts gain influence over endogenous directions of local agreement in \citet{Battaglini2002}, while \citet{levy2007limits} show how cross-issue spillovers can obstruct that logic for a single sender. \citet{chakraborty2007comparative} study comparative statements across ordered issues, and \citet{Chakraborty2010} obtain endogenous directions that balance a sender's exaggeration incentives across dimensions. Most directly related, \citet{hancart2025simple} study common-interest communication in which a two-dimensional state is summarized by a one-dimensional score. With a Gaussian prior, their equilibrium linear scores are projections onto the eigenvectors of the covariance-adjusted common loss matrix---equivalently given their bivariate setting, the ex ante best and worst linear scores. We instead characterize all rank-$k$ equilibrium action rules under differing priorities: the sender's loss matrix shapes the equilibrium set, whereas the receiver's does not.

Our dimensional restriction also differs from other technological constraints on communication. With common interest, \citet{jager2011voronoi} and \citet{rodriguez2026language} study finite languages in multidimensional environments. With sender commitment and possible misalignment, \citet{leTreustTomala2019capacity} bound information-theoretic capacity, while \citet{aybasTurkel2026coarse} fix the number of messages. Work on limited verification instead specifies which primitive facts can be checked \citep{glazer2004optimal,carroll2019strategic,antic2024selected,ballGao2026checking}. In our model, each message coordinate can carry arbitrarily fine information; the restriction is on the dimension of the receiver's continuous response, and equilibrium determines which linear combinations become issues.

Our low-dimensional assumption on receiver behavior connects as well to rational inattention, decision simplification, and organizational design. Most closely within this strand, \citet{hu2020multidimensional} studies communication of correlated multidimensional information to a rationally inattentive agent when the players agree on the appropriate state-contingent actions but differ in the relative importance they assign to different dimensions. There, the principal controls the information provided to the agent; here, communication is cheap talk, and reporting incentives determine which linear combinations of the state can constitute separate issues. More generally, information can be suppressed to redirect attention \citep{lipnowski2020attention}, and persuasion changes when receivers may disregard costly information \citep{BloedelSegal2021Persuading}. Scarce attention can focus coordination on a few tasks \citep{dessein2016rational}, while sparse attention and mental budgeting simplify individual choice \citep{gabaix2014sparsity,kHoszegi2020choice}. Related organizational work studies communication networks, knowledge hierarchies, task representations, technical languages, codes, and modular organization \citep{bolton1994firm,garicano2000hierarchies,cronin2007representational,cremer2007language,wernerfelt2004organizational,kocak2022separated,matouschek2025organizing}. Our players already share labels and coordinates; the friction is that only some groupings can be treated as separate issues without introducing misreporting incentives.

Finally, our paper relates to work on how conventional or natural meanings discipline strategic communication
\citep{farrell1993meaning,rabin1994model,reny2025natural,blume2025meaning}.
Whereas this literature begins with some pre-existing language and studies how its meanings constrain its strategic use and interpretation, our message coordinates have no exogenous substantive labels.
The receiver's linear response gives each coordinate an action interpretation, whereas equilibrium determines the state summary statistic it conveys; thus, the issues that message coordinates come to represent are themselves equilibrium objects.

\section{Endogenous issues
}\label{sec: model}
A receiver (R, he) consults with a sender (S, she) about what action $\actr \in \mathbb{R}^{\bign}$ he should take. R's optimal action is summarized by a stochastic state $\state \in \mathbb{R}^{\bign}$. S agrees it is best that R's action matches $\state$ as closely as possible. However, the two players disagree about how they should trade off the errors in the different dimensions, and so may differ in the way they rank suboptimal actions. Specifically, S and R suffer a loss that is quadratic in the difference between $\actr$ and $\state$, with each player's loss having different weights---in other words, the respective payoffs of S and R are
\[ \us=\state^\top \slo\state- (\state-\act)^\top \slo(\state-\act)
\quad \text{ and }\quad
\ur=\state^\top \rlo\state- (\state-\act)^\top \rlo(\state-\act),\]
where $\slo,\rlo\in\matnn$ are positive definite, and $\act$ is the (endogenous) random variable representing R's action. Note the above payoffs include the strategically irrelevant terms $\state^\top \slo \state$ and $\state^\top \rlo \state$ (normalizing the payoff from action zero to zero) to simplify payoff expressions.

The distribution of $\state$ has finite mean and variance. Essentially without further loss of generality, we take $\state$ to have a full-rank covariance matrix $\mvar$ (otherwise, one can restrict attention to a lower dimensional subspace), and zero mean (see \autoref{sec: disc}).

We view each action $\actr$ as representing a myriad of activities that can be decomposed and combined in various ways formalized via vector-space operations. Thus, each action can be separate into $\bign$ different components: 
given a basis $\{\basis_1,\ldots,\basis_\bign\}$, one can uniquely decompose every action $\actr \in \mathbb{R}^\bign$ into a sum $\actr = \sum_{i} \scalar_{i} \basis_{i}$ for some collection of scalars $\scalar_i \in \mathbb{R}$. We interpret a basis as a way of dividing R's action into separate issues, where $\actr = \sum_{i} \scalar_{i} \basis_{i}$ means R takes action $\scalar_i$ on issue $\basis_i$. 

The key novelty of our model is the assumption that players can only communicate about a small number $\litn \in \{0,\ldots,\bign\}$ of endogenously chosen issues.  To formalize this restriction, we use the strategic form of a standard cheap talk game \citep{Crawford1982,Green2007two} as a baseline. Thus, S sees $\state$ and chooses a message $\messr$ to send to R, who then takes an action. However, we modify this strategic form in two ways. First, S's message consists of $\litn$ real numbers, meaning $\messr \in \real^\litn$. And second---and more importantly---we restrict R's strategy space: instead of allowing R to use all maps from messages to actions, we restrict R to using a map that is \emph{linear} and \emph{full-rank}. Consequently, R's strategy is describable via a full-rank matrix, $\rstrat \in \rstrats$. Hence, if S sends message $\messr$ and R uses strategy $\rstrat$, then R's action equals $\actr = \rstrat \messr$. 

R's strategy space expresses a type of bounded rationality. In particular, R must interpret S's messages as a recommendation as to what to do on $\litn$ issues. However, R can choose which $\litn$ issues he interprets S as actually discussing. Thus, by choosing a full-rank matrix $\rstrat=[\rstratcol_1, \ldots, \rstratcol_\litn]$, R is deciding to interpret every S message $\messr$ as telling him to take action $\messr_i$ on issue $\rstratcol_i$, resulting in the aggregate action $\actr = \rstrat \messr = \sum_{i=1}^{\litn} \messr_i\rstratcol_i$.

Hence, a strategy profile is a pair, $(\sstrat,\rstrat)$, where $\sstrat:\real^\bign \rightarrow \real^\litn$ describes S's chosen message as a function of the state, and $\rstrat \in \rstrats$ has full rank. Let $\Rstrat$ be the set of all full-rank $\bign\times\litn$ real-matrices. A profile $(\sstrat,\rstrat)$ is a pure-strategy \textbf{Bayes-Nash equilibrium (BNE)} if and only if two conditions hold:\footnote{One can show the restriction to pure strategies is without loss.}
\begin{enumerate}    \item S's message is optimal given $\state$: 
    \[ 
    \sstrat(\stater) \in \argmin_{\messr \in \reallit}\brac{(\rstrat\messr - \stater)^\top \slo (\rstrat\messr - \stater)} \quad\forall\ \stater\in\real^\bign, \label{eq: S-BR} \tag{S-BR}
    \]
    \item R's strategy is optimal:\footnote{Whenever S's incentive constraint is satisfied, it follows from our analysis that this R expected value is well defined for any strategy he chooses.} 
    \[
    \rstrat \in \argmin_{\rstrat \in \Rstrat}\bbE\curlyb{\brac{\rstrat\sstrat(\state) - \state}^\top \rlo \brac{\rstrat\sstrat(\state) - \state}}. \label{eq: R-BR} \tag{R-BR}
    \]
\end{enumerate}
In what follows, we use the term \textbf{equilibrium} as a short-hand for pure-strategy BNE.

Our model shares some limitations with the standard cheap talk game. In particular, R must choose how to interpret S's messages before they are realized. Consequently, S has no way of alerting R that a certain issue is more urgent than the others. In other words, S cannot send R a message that changes the issues he expects S to talk about. Instead, the issues S talks about are determined endogenously through R's expectations and the players' equilibrium behaviors.

One limitation of the usual cheap talk game that we avoid is the possibility of babbling. The reason is that we assume R uses the entire capacity of his communication channel---i.e., that R's linear map must have full rank. Under this assumption, when $\litn\geq 1$, R's action is always responsive to S's message and some information is always communicated in equilibrium.\footnote{We discuss the full-rank assumption in greater depth in \autoref{sec: disc}.}

Thus, the only way to get babbling in our model is to have $\litn=0$. In this case, the players cannot do anything but babble, and R always chooses $0$. At the opposite extreme, when $\litn=\bign$, equilibrium communication is perfect: every R strategy has full range, and so S can always induce R to match the state perfectly. We therefore restrict the capacity $\litn$ to be strictly between $0$ and $\bign$ in what follows, except for \autoref{sec: investment}, where we endogenize~$\litn$. 

Before proceeding, we note there are two other game forms whose solutions coincide with the equilibria we analyze in the sequel. The first equivalent model is a standard cheap talk model with the usual strategy space for R, but where we restrict the set of allowable equilibria. Specifically, R is free to choose his action after seeing S's message in $\real^\litn$. However, we consider only pure-strategy perfect Bayesian equilibria in which R's strategy happens to be a full-rank linear map. For such equilibria to exist and coincide with the equilibria of our model, one must impose the additional assumption that $\state$ has an elliptical distribution (e.g., Gaussian).\footnote{One may be interested in requiring also that S's equilibrium strategy is a full-rank linear map. This property turns out to hold automatically given the  structure of R's strategy.}

The second equivalent game form is one in which both agents choose linear maps. Here, S and R move simultaneously without any knowledge of the state's realization, with S choosing a full-rank matrix $\sstrat\in\matkn$ and R choosing a full-rank matrix $\rstrat\in\matnk$, resulting in action $\act:=\rstrat\sstrat\state$. The solution concept is pure-strategy Nash equilibrium.

In both of the above formulations, the characterization of equilibrium action rules (and hence payoffs) is exactly the same as in our main model. We therefore welcome a reader to interpret our model in any of these three equivalent ways.

\section{An illustrative example}\label{sec: example}

This section illustrates the key tensions in our model using a simple example. To give the example some color, suppose R is a school board making decisions on two dimensions: the school's curriculum and teacher staffing. The head of the teacher's union, S, knows the optimal decision on both of these dimensions. To simplify the analysis, suppose the pair of optimal decision on  these two dimensions is Gaussian, with each dimension having a mean of $0$ and a variance of $1$. More importantly, the two dimensions are positively correlated, with correlation $\excor \in (0,1)$. Note that we do not require $\state$ to be Gaussian in the general model.

We assume R cares more about choosing the correct curriculum, while S cares more about teacher staffing. More concretely, 
\[
\slo =\twomat{\epsilon}{0}{0}{100}\quad \text{and} \quad \rlo = \twomat{98}002,
\]
for some $\epsilon>0$. That is, S and R's payoffs are respectively given by
\[
\us = \state^\top\slo\state - \epsilon (\act_1 - \state_1)^2 - 100(\act_2-\state_2)^2
\quad \text{and} \quad 
\ur = \state^\top \rlo \state -98(\act_1-\state_1)^2-2(\act_2-\state_2)^2,
\]
where dimension 1 represents the choice of curriculum (e.g., the relative emphasis on sciences vs. the arts), and dimension 2 represents changes in teacher staffing. We caricature this situation by considering the case where $\epsilon\approx 0$, and describe the equilibria in the limit as $\epsilon$ vanishes.\footnote{More precisely, we will characterize the maps from state to action that arise as the pointwise limit of equilibria as $\epsilon\searrow0$.}

Due to various attention and time constraints, R can consult with S about only one issue, that is $\litn=1$. Our main interest is in understanding which issue S ends up communicating in equilibrium.

If S could have her way, she would use R's attention to ensure he takes the optimal staffing decision. One can create such an equilibrium by having S tell R the value of that dimension, $\mess = \state_2$. For an explanation, note that having seen such a message, R is clearly best choosing the correct teacher staffing, $\act_2 = \mess =\state_2$. Moreover, the correlation between the two dimensions enables R to make inferences about the optimal curriculum.  In particular, the example's Gaussian assumption means $\bbE[\state_1|\mess]=\excor\state_2$. Therefore, by setting $\rstrat = [\excor, 1]^\top$, R can ensure his optimal action matches the conditional expectation of $\state$ given $\mess$---which is optimal. R's payoff is $100-98(1-\excor^2)$, while S gets her first-best payoff of $100$. 

Unlike the optimal teacher staffing, S cannot credibly tell R the optimal curriculum---that is, $\mess = \state_1$ cannot occur in equilibrium. To see why, observe that a symmetric argument to the one above implies that R's best reply to $\mess = \state_1$ is $\rstrat = [1, \excor]^\top$. Hence, $\act_1 = \mess$, and $\act_2 = \excor\mess$. Given this R strategy, S's best reply is not to send $\mess =
\state_1$. Rather, she prefers to send $\mess = \state_2/\excor$: R would then take $\act_2 = \state_2$, giving S her first-best payoff. This deviation illustrates our environment's key communication friction: S cannot credibly communicate about issues that mix topics she prioritizes unevenly.  

To have S communicate about the optimal curriculum, S must clean out the information it provides about teacher staffing. More concretely, suppose S sent the message $\mess = \state_1 - \excor \state_2$. In this case, R's conditional expectations are given by $\bbE[\state_1|\mess] = \mess$, and $\bbE[\state_2|\mess] = 0$, and so setting $\rstrat = [1 , 0]^\top$ is optimal. Moreover, S cannot benefit from deviating: her utility is zero regardless of her message. R's payoff in this equilibrium is $98(1-\excor^2)$.

These two equilibria demonstrate a more general phenomenon: S can communicate an issue to R in equilibrium if and only if that issue combines only dimensions to which S assigns the same value. Whenever an issue mixes topics that S values unevenly, S benefits from slanting communication in favor of the topics that matter more to her. Conversely, S never gains from miscommunicating an issue that influences R's actions on dimensions that concern S equally: such miscommunication can only lead R to take a wrong action on said issues, thereby reducing S's payoff. Moreover, if R expects S to honestly communicate an issue, then R should always take S's report at face value, making accurately communicating said issue self-enforcing.   

The above intuition also suggests that R's priorities play no role in shaping equilibrium communication. Broadly speaking, the reason is that, once R knows S's communication is honest, R faces a simple signal extraction problem, the solution to which is the best linear estimator of the state given S's report. Since this estimator does not depend on R's preferences, neither does the equilibrium set. One can see this independence in the above example: the equilibria we derived remain even if we changed R's loss matrix, since that matrix has no impact on his best replies.\footnote{While the argument we presented in the example relied on the example's Gaussian assumption, we show this assumption is not necessary in general.}

In the sequel, we show that equilibrium communication is shaped by S's preferences and the state's variance-covariance matrix. In particular, equilibria are shaped by an $\slo$-perspective, which is a decomposition of R's actions into issues that S views as unrelated. The next section defines this object formally and uses it to characterize the model's equilibrium set. 

\section{What can you talk about?}\label{sec: eqm char}

This section presents our equilibrium characterization. Our characterization shows  the equilibrium is shaped by S's perspective---an object which we will define shortly. Intuitively speaking, an $\slo$-perspective provides a way of framing R's decision as $\bign$ separate choices that S views as unrelated. This perspective depends only on the state's variance-covariance matrix $\mvar$ and S's loss matrix $\slo$. As a consequence, we show R's priorities are irrelevant in equilibrium---changing R's loss matrix has no impact on the equilibrium set.

We begin by defining a perspective. Given an $\bign\times\bign$ matrix $\View = [\view_1,\ldots,\view_\litn]$ and an index $i$, we let $\state^{\View}_{i}:=\view_i \cdot \state$ and $\act^{\View}_{i}:=\view_i \cdot \act$ respectively be the (non-normalized) projection of $\state$ and $\act$ onto the vector $\view_i$. We say that $\View$ is a \textbf{perspective} if $\state^{\View}_{i}$ and $\state^{\View}_{j}$ are uncorrelated for all $i\neq j$, and if $\state^{\View}_i$ has unit variance for every $i$. Intuitively, a perspective decomposes R's decision problem into $\bign$ statistically unrelated choices. Geometrically, $\View$ is a perspective if $(\sqrt{\mvar}\view_i)_{i =1}^{\bign}$ is an \emph{orthonormal basis} for $\real^\bign$; that is, the vectors $(\sqrt{\mvar}\view_i)_{i =1}^{\bign}$ have unit length and are pairwise orthogonal (hence linearly independent). 

A perspective $\View$ is an \textbf{$\slo$-perspective} if some weights $\seig_1,\ldots,\seig_\bign \in \real_{++}$ exist such that every $\stater,\actr\in\real^\bign$ have
\begin{equation*}
    (\actr - \stater)^{\top}\slo (\actr - \stater) = \sum_{i=1}^{\bign} \seig_i (\actr^{\View}_i - \stater^{\View}_i)^2.
\end{equation*}
Thus, an $\slo$-perspective separates R's action into $\bign$ decisions that S views as unrelated to each other---both statistically and in the sense of payoff separability. The weights $\seig_1,\ldots \seig_\bign$ designate how much S cares about each issue. The order and directionality of the issues is irrelevant: we can turn one $\slo$-perspective into a different one by permuting columns and multiplying some of them by $-1$. We refer to two $\slo$-perspectives that are related in this way as \textbf{issue equivalent}. If  all $\slo$-perspectives are issue equivalent, we say the $\slo$-perspective is \textbf{unique}.
Geometrically, an $\slo$-perspective is an eigendecomposition of $\sqrt{\mvar}\slo\sqrt{\mvar}$. As such, an $\slo$-perspective always exists, and is unique whenever all $\bign$ weights are distinct---a generic property.\footnote{See \autoref{def: genericity} for a formal definition of genericity, and \autoref{lem: generically distinct weights} for a proof of genericity.} Moreover, these weights are unique up to the same permutation as the corresponding columns.

We now demonstrate these definitions in the context of another example.

\begin{example}\label{ex: CEO-manager}
    Suppose R is a CEO of a company, and S is a division manager. S's division has two products. R's decision is how much money to dedicate to each product (net of some status quo), where $\act_i$ is the amount R allocates to product $i$. S knows the optimal allocation $(\state_1,\state_2)$, which is stochastic and has a unit variance-covariance matrix $\mvar = I$. Both S and R would like to give the right amount of resources to each product, but differ in their priorities, as expressed by the different loss matrices:
    \[
    \slo =\twomat{3}{1}{1}{3} ,\ \ \ \ \ \  \rlo = \twomat{98}002.
    \]
    The matrix representation makes R's priorities quite clear: R cares mostly about product 1 (e.g., because it is R's pet project). Indeed, writing R's utility more explicitly gives
    \[
    \ur = \state^{\top}\rlo\state- 98(\act_1 - \state_1)^2 -2(\act_2 - \state_2)^2.
    \]
    By contrast, S's priorities are harder to discern from the matrix representation. Expanding the matrix term is not very helpful either:
    \[
    \us = \state^{\top}\slo \state - 3 (\act_1 - \state_1)^2 - 3 (\act_2 - \state_2)^2 -2 (\act_1 - \state_1) (\act_2 - \state_2).
    \]
    One can get a clearer picture of S's priorities, however, by finding her perspective. Using the above definition reveals that S has a unique perspective,
    \[
    \View = \tfrac{1}{\sqrt{2}} \left[\twovec{1}{1} \twovec{1}{-1} \right], \quad \text{and}\quad (\seig_1,\seig_2)=(4,2).
    \]
    Accordingly, one can write S's loss as
    \[
    \us = \state' \slo \state - 2\left[(\act_1 + \act_2) - (\state_1 + \state_2)\right]^2 - \left[(\act_1 - \act_2) - (\state_1 - \state_2)\right]^2. 
    \]
    These expressions reveal S's true priorities: S's cares more about the total budget for her division than she cares about the way that budget is allocated.  \qed
\end{example}

To get a sense of why S's perspective is central for equilibrium, consider S's problem \eqref{eq: S-BR} for a fixed state realization $\stater \in \real^\bign$. This problem has a strictly concave and differentiable objective, and so its solution is characterized by a first order condition. This condition describes a tangency between two objects. The first object is the range of R's strategy $\rstrat$, which is a $\litn$-dimensional linear subspace of $\mathbb{R}^\bign$. This range represents S's ``budget'' set: S can induce any action in this range by choosing an appropriate message. The second object is an S indifference curve. These curves are ellipsoids centered at $\stater$. S's first-order condition says that inducing R to take $\actr$ is optimal for S only if $\actr$ is the tangency point between its indifference curve and the range of $\rstrat$ (see \autoref{fig:S tangency}).

\begin{figure}
    \centering
    \includegraphics[width=0.5\linewidth]{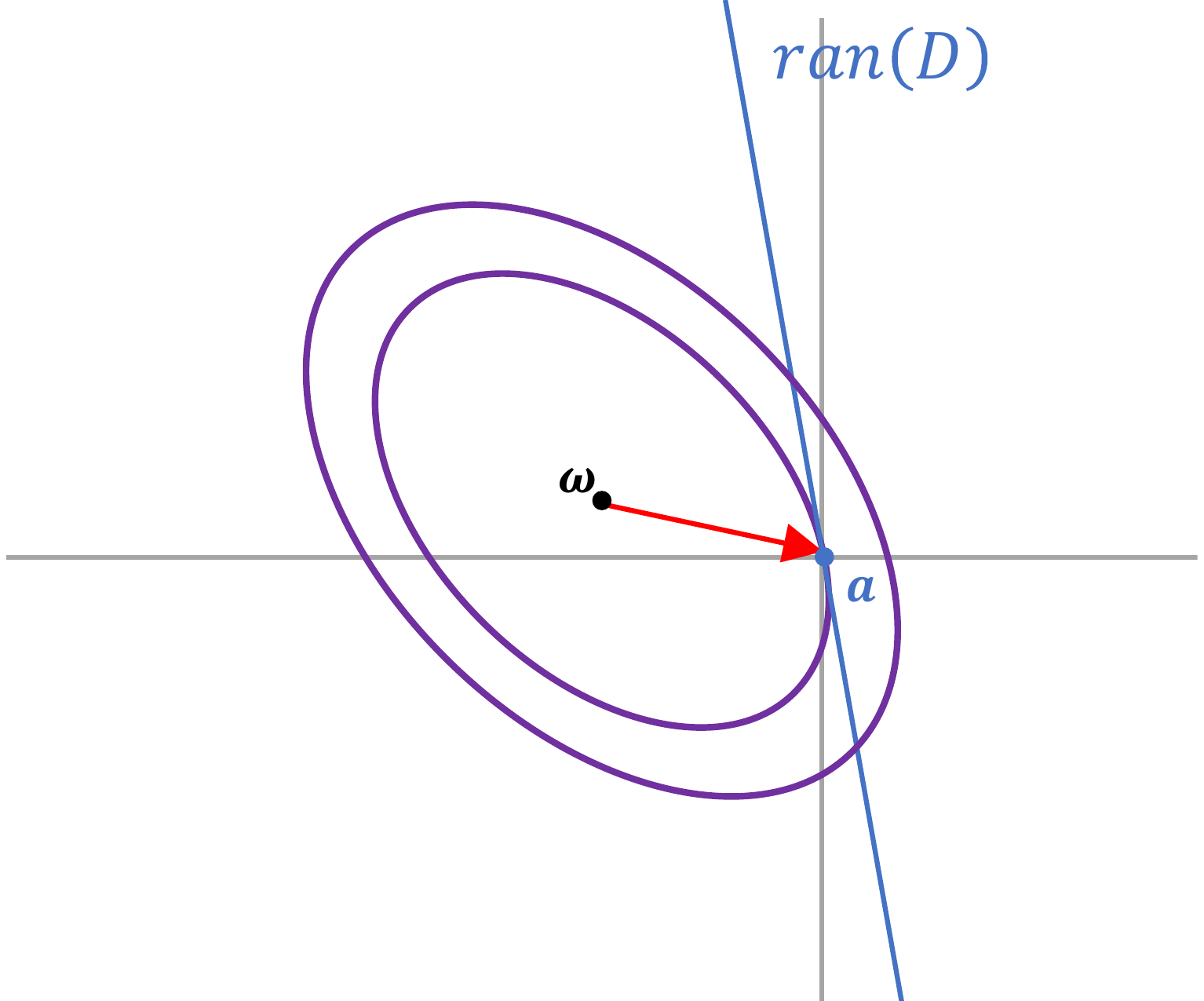}
    \caption{Geometric illustration of the tangency conditions characterizing S-optimality when $\bign=2$ and $\litn=1$.}
    \label{fig:S tangency}
\end{figure}

S's perspective characterizes an orthogonality between the linear space tangent to S's indifference curve at an action $\actr$ and the indifference curve's radius vector that goes to $\actr$---i.e., the line connecting $\actr$ to the realized state $\stater$. When $\mvar$ is the identity, one can describe this orthogonality using the standard vector dot product. In particular, the radius vector to $\actr$ is perpendicular to S's indifference curve's tangent subspace at $\actr$ if and only if that subspace admits a subset of S's perspective as a basis. \autoref{fig:radius-tangent} demonstrates this tangency in the context of \autoref{ex: CEO-manager}.

\begin{figure}
    \centering
    \includegraphics[width=0.5\linewidth]{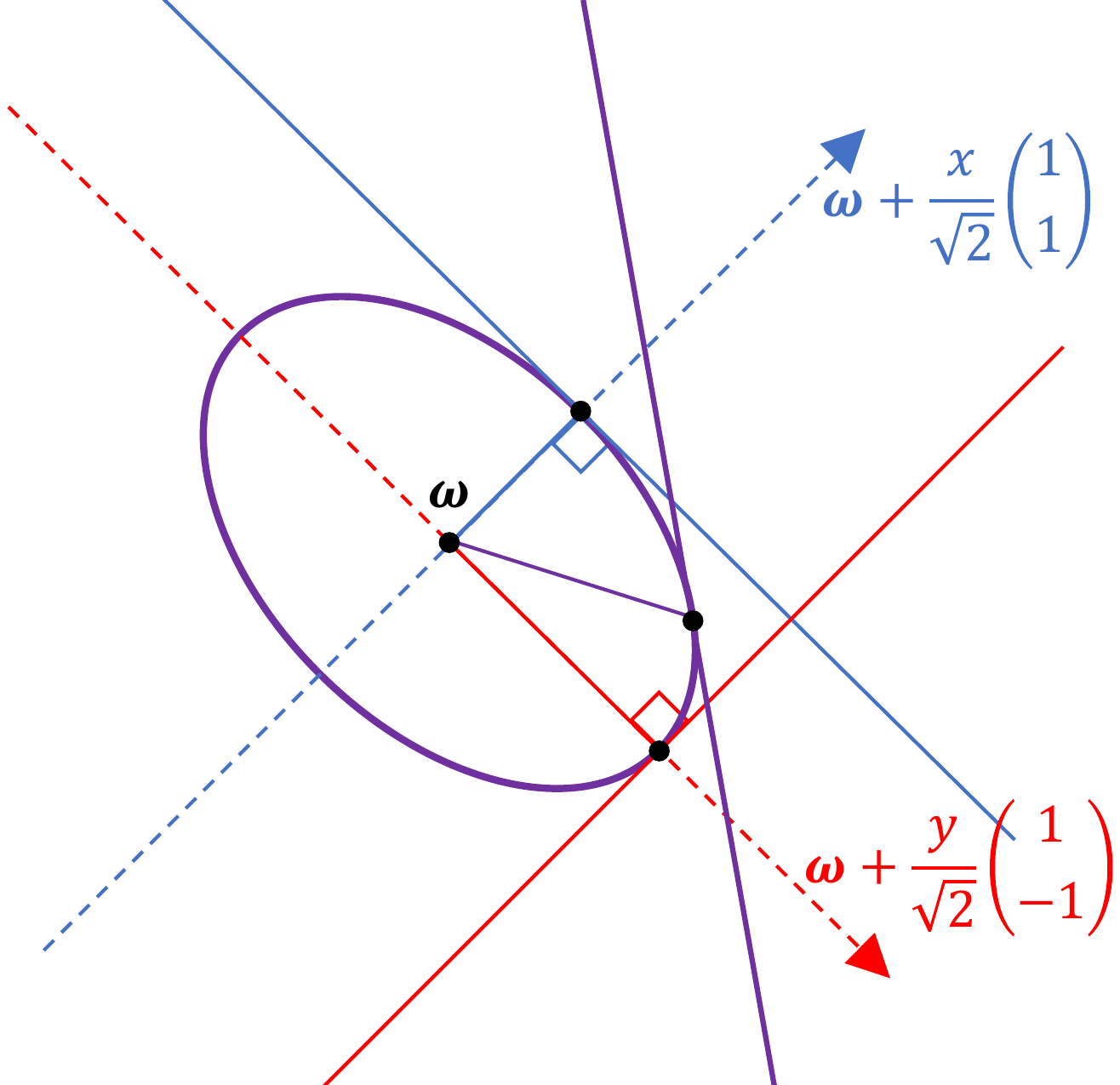}
    \caption{Illustration of the relationship between S's perspective and orthogonality in \autoref{ex: CEO-manager}. Dashed lines represent shifts in the state in the directions of S's perspective. These directions coincide with the direction of the radius vector to a point on the indifference curve exactly when the vector is perpendicular to the tangent at the point.}
    \label{fig:radius-tangent}
\end{figure}

Intuitively, the radius vector to $\actr$ represents the unique non-compensable direction: no local movement of $\actr$ would compensate S for a local push of $\actr$ away from $\stater$ in the radius's direction. By contrast, the tangent hyperplane represents all the directions of indifference---that is, all directions to which one can perturb $\actr$ without meaningfully changing S's losses. The previously mentioned  orthogonality means that the indifference directions and the non-compensable direction are separable from each other both in a statistical sense and from the point of view of S's preferences. 

As our next result shows, S's perspective shapes equilibrium communication. To present the result formally, say that a strategy profile $(\sstrat,\rstrat)$ \textbf{communicates an $\slo$-perspective} if an $\slo$-perspective $\View$ exists such that $\mess = \sstrat(\state)$ and $\act = \rstrat \mess$ have\footnote{Every perspective admits some strategy profile that communicates it---see \autoref{sec: proofs} for an explicit description.}
\begin{equation*}
\begin{split}
\mess_{i} & = \state^{\View}_{i} \hspace{0.16cm} \text{ for }i=1,\ldots,\litn, \\
\act^{\View}_{i} & = \begin{cases}
	\mess_{i} & \text{if }i=1,\ldots,\litn,\\ 
	0 & \text{if }i=\litn+1,\ldots,\bign.
\end{cases}
\end{split}
\end{equation*}
In words, a profile communicates an $\slo$-perspective $\View$ if S tells R what to do on the first $\litn$ issues of $\View$; R responds by doing as he is told on those $\litn$ issues, and taking the (prior-optimal) zero action on the remaining issues. We use the convention that communication is restricted to the first $\litn$ issues of $\View$. 
Hence, two profiles that communicate issue-equivalent $\slo$-perspectives typically transmit different information. 

As a demonstration, consider \autoref{ex: CEO-manager}. In this example, there are essentially only two strategy profiles that communicate S's perspective. We illustrate these profiles in \autoref{fig: both equilibria}.  
In the first profile, $(\sstrat,\rstrat)$, the sender tells the receiver the size of $\state$'s projection onto $\tfrac{1}{\sqrt{2}}(1,1)^\top$. The second profile $(\sstratb,\rstratb)$ has S reporting the size of the projection of $\state$ onto $\tfrac{1}{\sqrt{2}}(1,-1)^\top$. These correspond to the respective S-strategies $\sstrat(\state) = \tfrac{1}{\sqrt{2}}(\state_1 + \state_2)$ and $\sstratb(\state) = \tfrac{1}{\sqrt{2}}(\state_1 - \state_2)$. In both profiles, R responds by matching the size of the projection to the appropriate vector, and setting the projection onto the other vector to zero. More precisely, $\rstrat_1\mess = \rstrat_2\mess= \tfrac{1}{\sqrt2}\mess$, and $\rstratb_1\mess = - \rstratb_2\mess= \tfrac{1}{\sqrt2}\mess.$ Thus, the first profile is equivalent to the division manager communicating the appropriate total budget for the division and the CEO splitting that budget evenly between the two products. The second profile is equivalent to the division manager telling the CEO the correct difference in the ideal budget for the division's two products, and the CEO splitting that difference evenly.

\begin{figure}[htbp]
    \centering

    \begin{subfigure}[t]{0.48\textwidth}
        \centering
        \includegraphics[width=\linewidth]{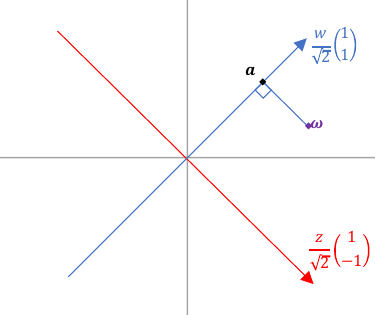}
        \caption{S dictates the total budget, and R splits it equally between the two products}
        \label{fig:left}
    \end{subfigure}
    \hfill
    \begin{subfigure}[t]{0.48\textwidth}
        \centering
        \includegraphics[width=\linewidth]{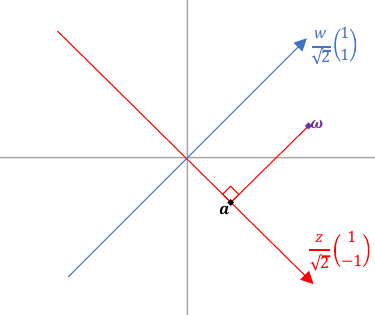}
        \caption{S dictates the relative budget, and R uses the ex-ante correct total budget.}
        \label{fig:right}
    \end{subfigure}

    \caption{The two equilibrium state-contingent action choices for \autoref{ex: CEO-manager}.}
    \label{fig: both equilibria}
\end{figure}

\autoref{prop: eqm char} says the game's equilibria are characterized by the profiles that communicate S's perspective. To state this result, say that two strategy profiles are \textbf{equivalent} if they induce the same joint distribution for $\state$ and $\act$.

\begin{proposition}\label{prop: eqm char} \ 
\begin{enumerate}[(i)]
    \item Any strategy profile that communicates an $\slo$-perspective is an equilibrium.
    \item Any equilibrium is equivalent to one that communicates an $\slo$-perspective.
\end{enumerate}
\end{proposition}

The proposition highlights how equilibrium communication is shaped by two features of the environment: S's priorities and the state's stochastic structure. Intuitively, S's incentives requires S to view the issues on which R acts  as separate from the issues about which he does not act---otherwise, S would distort her messages so as to better tailor R's actions on issues she assigns a higher weight. Meanwhile, R chooses his issue interpretation to make S's communication maximally instrumental; hence, the issues on which R acts are uncorrelated with the issues he leaves idle.

For a more detailed overview of the proposition's logic, consider again S's problem given an arbitrary state realization $\stater\in \real^\bign$. As explained earlier, the problem's solution is characterized by tangency between S's indifference curve and the range of R's strategy. This tangency is defined by a linear equation, and so implies S's best response is a linear function of the state's realization $\stater$. Geometrically, this linearity is most transparent when $\litn=\bign-1$. In that case, one can decompose any change in the state's realization $\stater$ into two linear operations. The first operation is a translation in a direction parallel to the range of $\rstrat$. Such a translation results in an identical translation of the resulting tangency point. The second operation is a shift toward or away from the range of $\rstrat$ in the direction of the radius vector of S's indifference curve---a shift that leaves the resulting tangency point unchanged.

Consider now R's best reply. For intuition, suppose the states are uncorrelated and have unit variance, $\mvar=I$. In equilibrium, R expects S's messaging strategy $\sstrat$ to be a linear map. One can therefore decompose the state into two parts: its projection onto the map's null space $\statenull$, and its projection $\staterow$ onto the null's orthogonal complement, also known as the map's row space. Moreover, S's strategy is invertible when restricted to this row space. It turns out that R's best reply inverts this map, matching $\staterow$ exactly and estimating $\statenull$ using its best linear predictor given $\staterow$. Intuitively, R's action is uncorrelated with $\statenull$ no matter what strategy he chooses, and hence perfectly matching $\staterow$ comes at no opportunity cost.\footnote{When $\mvar\neq I$, an analogous decomposition applies, using orthogonality with respect to a prior-normalized coordinate system.}

Geometrically, R optimality can be described via an orthogonality condition between the range of R's strategy and the line connecting each state realization to its associated action. When $\mvar=I$, we can express this orthogonality via the usual dot product: the line connecting the state to the action must be perpendicular to the range of R's strategy. In fact, in this case, since the best linear predictor of $\statenull$ given $\staterow$ is zero, the range of R's strategy and the row space of S's message strategy must coincide. \autoref{fig:null-rowspace} illustrates this relationship.

\begin{figure}
    \centering
    \includegraphics[width=0.5\linewidth]{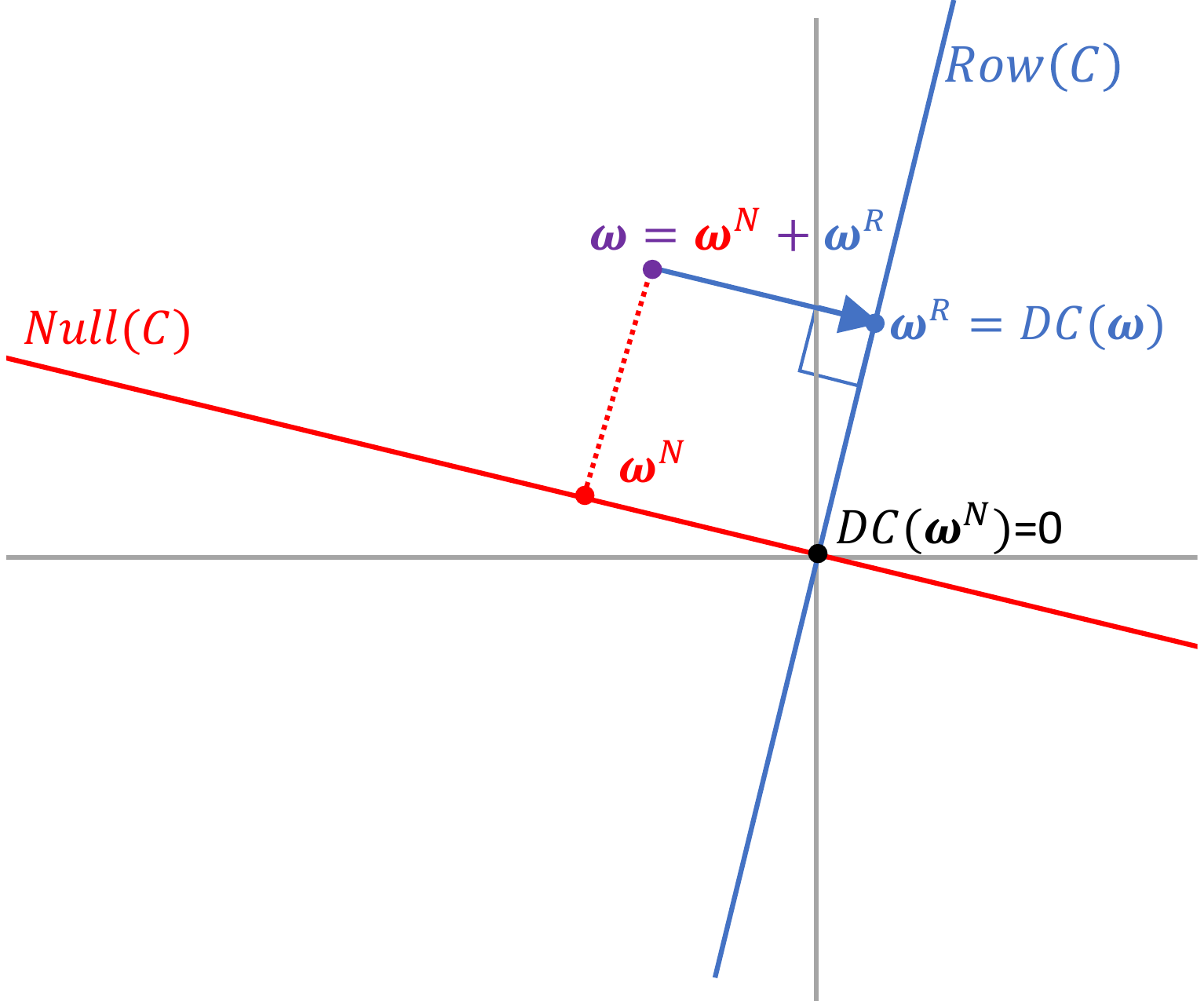}
    \caption{The orthogonality condition characterizing R optimality for $\bign=2$ and $\litn=1$, specialized to the case of $\mvar = I$. The figure illustrates the decomposition of $\state$ into  an element $\statenull$ of the null space of $\sstrat$ (denoted by $Null(\sstrat)$) and an element $\staterow$ of the row space of $\sstrat$ (denoted by $Row(\sstrat)$). In this setting, R optimality is equivalent to requiring the line connecting  $\state$ to $\staterow = \rstrat\sstrat(\state)$ to be perpendicular to the range of $\rstrat$. Consequently, the row space of $\sstrat$ and the range of $\rstrat$ coincide.}
    \label{fig:null-rowspace}
\end{figure}

The proposition follows from combining the above-mentioned geometric facts. By R incentives, the line connecting a state realization $\stater$ to the resulting action $\actr$ must be orthogonal to R's action space. This space, in turn, needs to be tangent to the S-indifference curve going through $\actr$. All that remains is to recall the fact mentioned above: a radius vector of an S-indifference curve is  orthogonal to a tangent linear subspace if and only if that space admits a subset of an $\slo$-perspective as a basis. Replicating equilibrium communication using this basis delivers an equivalent profile that communicates S's perspective. 

\section{What do you want to talk about?}\label{sec: eqm payoffs}
By \autoref{prop: eqm char}, communication must decompose R's decision according to an $\slo$-perspective. Subject to this decomposition, however, S can communicate to R about any collection of issues. This section asks: which issues would each player want to talk about?

For a more formal exposition of this section's question, note that an $\slo$-perspective $\View$ specifies two things. First, the columns of $\View$ specify a particular way of framing R's decision; that is, decomposing it into separate issues. Second, the order of the columns specifies which issues are discussed in a $\View$-equilibrium. Our interest in this section is in the second component; that is, which of the issues in $\View$ would R like S to communicate, and which would S prefer to communicate to R?

To answer our question, we calculate the players' payoffs from a given equilibrium. To present the calculation's outcome, let $\Viewb$ be an $\rlo$-perspective with associated weights $\reig_1,\ldots,\reig_\bign$, and $\View$ be an $\slo$-perspective with associated weights $\seig_1,\ldots,\seig_\bign$. For $i,j=1,\ldots\bign$, define the \textbf{coefficient of determination} (also known as the ``$r^2$'') between $\state^{\View}_i$ and $\state^{\Viewb}_j$ as
\[
\codet{\state^{\View}_i}{\state^{\Viewb}_j}:=\frac{\left(\bbE[\state^{\View}_i\state^{\Viewb}_j]\right)^2}{\var(\state^{\View}_i)\var(\state^{\Viewb}_j)} = \left(\bbE[\state^{\View}_i\state^{\Viewb}_j]\right)^2.
\]
In words, $\codet{\state^{\View}_i}{\state^{\Viewb}_j}$ is the percentage of $\state^{\Viewb}_j$'s variance one can capture if one were to use the best linear estimate of $\state^{\Viewb}_j$ given $\state^{\View}_i$ (or vice versa). We interpret $\codet{\state^{\View}_i}{\state^{\Viewb}_j}$ as a measure of informativeness: the higher $\codet{\state^{\View}_i}{\state^{\Viewb}_j}$ is, the more information issue $\view_i$ provides about issue $\viewb_j$ (and vice versa).

\begin{proposition}\label{prop: eqlbm payoffs}
    Fix an $\slo$-perspective $\View$ with weights $\seig_1,\ldots,\seig_\bign$, and an $\rlo$-perspective $\Viewb$ with weights $\reig_1,\ldots,\reig_\bign$. S and R's payoffs in any equilibrium that communicates $\View$ are, respectively,
\begin{equation*}\label{eq: eqlbm payoffs}
        \bbE[\us] = \sum_{i=1}^{\litn}\seig_i \quad \text{ and }\quad \bbE[\ur]=\sum_{i=1}^\litn\sum_{j=1}^\bign \reig_j \codet{\state^{\View}_i}{\state^{\Viewb}_{j}}.
    \end{equation*}
\end{proposition}

The result highlights the way each player evaluates the different issues. S's valuations are straightforward: the benefit to S from communicating issue $\view_i$ to R is $\seig_i$. Thus, the weight associated with issue in $\View$ gives that issue's value to S. 

R's evaluation of different issues is more nuanced. If S communicates to R the issue $\view_i$, R obtains a value of
\begin{equation}\label{eq: R issue payoffs}
    \issur_{i} := \sum_{j=1}^\bign \reig_j \codet{\state^{\View}_{i}}{\state^{\Viewb}_{j}}.
\end{equation}
In other words, R prefers S to communicate issues that are provide more information about issues R views as more important. This preference is easiest to see when $\View=\Viewb$: in this case, R's benefit from learning issue $\view_j=\viewb_j$ is just $\reig_j$, and so R would like to order $\View$ according to his weights.  

Next, we illustrate \autoref{prop: eqlbm payoffs} in the context of \autoref{ex: CEO-manager}.
\addtocounter{example}{-1}
\begin{example}[continued] 
Recall this example has $\mvar=I$ and loss matrices
    \[
    \slo =\twomat{3}{1}{1}{3} \text{ and } \rlo = \twomat{98}002.
    \]
    The corresponding perspectives are
    \[
    \View = \tfrac{1}{\sqrt{2}} \left[\twovec{1}{1} \twovec{1}{-1} \right], (\seig_1,\seig_2)=(4,2) \quad \text{and} \quad \Viewb=\left[\twovec{1}{0} \twovec{0}{1} \right], (\reig_1,\reig_2)=(98,2).
    \]
    By \autoref{prop: eqm char}, this example admits at most two equilibria: one which communicates the issue $\view_1=(1,1)^{\top}$, and the other that communicates $\view_2=(-1,1)^{\top}$. If S communicates the first issue,  she can guarantees that R will allocate the correct total budget for her division, resulting in a payoff of $4$. This payoff is larger than the payoff of 2 that S would get if she were to communicate the second issue---i.e., the optimal difference in budget allocation between the two products. 

    It turns out that R does not care which of these two equilibria is played. To see this identity, note the coefficient of determination between the issues in $\View$ and in $\Viewb$ are $\codet{\state^{\View}_i}{\state^{\Viewb}_j}=0.5$ for any $i$ and $j$. Therefore, R's payoff is equal to $0.5(2)+0.5(98)=50$ in both equilibria. 

        What if instead R interacted with a sender with loss matrix $\rlo$? In this case, S and R would have the same perspective, and so every equilibrium is equivalent to S communicating either the issue $\viewb_1 = (1,0)^\top$ or the issue $\viewb_2 = (0,1)^\top$. S's and R's payoffs from each of these equilibria are the same: they both get a payoff of $\reig_1=98$ from the first equilibrium, and $\reig_2=2$ from the second.     \qed
\end{example}

\section{Who do you want to talk to?}\label{sec: which sender}

This section explores the implications of our equilibrium characterization for the question: who would you like to talk to? The answer is clearly relevant when one can actually choose one's conversation partner. But the answer can also be relevant for organizational design; by changing the design of an organization, one can impact the priorities of its members. For an illustration, consider the CEO in \autoref{ex: CEO-manager}. In addition to replacing the manager, the CEO can change the division manager's priorities (i.e., $\slo$) in many ways: he can change what products are under the manager's authority, restructure the manager's contract, change the way the division interacts with other parts of the organization, etc.

The results of this section highlight that while S does not care about $\rlo$, R's payoff vary with $\slo$ in a  
rich way. As previously noted, R evaluates a given $\slo$ according to the correlation between S's perspective and his own. Equilibrium selection turns out to play a critical role: whether R wants more or less correlation depends on whether equilibrium selection is favorable or adversarial. Consequently, the way R ranks senders depends on which equilibrium R expects to be played. In particular, we show there is a sense in which an S that is good under R-favorable equilibrium selection is often bad under R-adversarial equilibrium selection. A rough intuition is that an S whose perspective is ``synchronized'' with R's has issues that are highly correlated with R's most important issues and other issues that are highly correlated with R's least important issues.

The asymmetry between S and R is driven by the forces that shape equilibrium communication. In particular, R's preferences are irrelevant in equilibrium; R simply attempts to match the state as best as he can given his information. Therefore:

\begin{corollary}\label{cor: S cares less about R}
    S's equilibrium payoff set does not depend on $\rlo$. 
\end{corollary}

Whereas S's equilibrium payoff set is independent of R's loss matrix, R's equilibrium payoff set clearly depends on $\slo$---specifically, S's perspective. However, as previously mentioned, R's ranking of different $\slo$ depends on the equilibrium selection criterion. Indeed, we have already observed an instance of this dependency in \autoref{ex: CEO-manager}: the CEO would like to replace the division manager with a manager who shares his priorities only if he expects his favorite equilibrium to be played. By contrast, if the CEO expects his worst equilibrium to be played, he prefers the original manager.

Our next result outlines the way in which R's best and worst equilibrium payoffs vary with $\slo$. To state this result, given an $\slo$, take $\goodur(\slo,\rlo)$ and $\badur(\slo,\rlo)$ respectively to be R's highest and lowest equilibrium payoffs (which our analysis shows exist) when his loss matrix is $\rlo$ and S's matrix is $\slo$. Letting $(\reig_i)_{i=1}^\bign$ be R's weights associated with some $\rlo$-perspective, take $\lowurb$ to be the sum of the $\litn$ lowest weights, $\highurb$ be the sum of the $\litn$ highest weights, and $\medurb$ be $\litn$ times the average weight.

\begin{proposition}\label{prop: receiver payoff set}
For a pair $(\highur,\lowur)\in\real^2$, the following are equivalent. 
\begin{enumerate}[(i)]
    \item Some $\slo$ has $\goodur(\slo,\rlo)=\highur$ and $\badur(\slo,\rlo)=\lowur$.\label{prop: receiver payoff set, payoff part}
    \item\label{prop: receiver payoff set, defining inequalities} $\lowurb\leq\lowur\leq\medurb\leq\highur\leq\highurb$ and 
$\medurb=(1-\scalar)\lowur+\scalar\highur$ for some $\scalar\in\co\curlyb{\tfrac\litn\bign,\ 1-\tfrac\litn\bign}$.
\end{enumerate}
\end{proposition}
The proposition shows that an $\slo$ with an R-best equilibrium payoff of $\highur$ and a R-worst equilibrium payoff of $\lowur$ exist if and only if $\highur$ and $\lowur$ can be expressed as a moderate binary splitting of $\medurb$. This moderation is expressed both in terms of its support---$\lowur$ and $\highur$ must be between $\lowurb$ and $\highurb$--- and in terms of each point's mass: neither $\highur$ and $\lowur$ can have a mass outside $\co\curlyb{\tfrac\litn\bign,\ 1-\tfrac\litn\bign}$, which is an interval centered around $\tfrac{1}{2}$.

Underlying the proposition are two related constraints on the amount of information (as measured by $\codett$) the issues in S's perspective $\View$ can provide about issues in R's perspective $\Viewb$, and vice versa.
The first constraint is about the amount of information one can learn about $\Viewb$ by observing a single issue from $\View$. In particular, the more information $\view_i$ provides about $\viewb_j$, the less information $\view_i$ must provide about other issues in $\Viewb$. The reason is that knowing $(\state^{\Viewb}_m)_{m=1}^{\bign}$ is equivalent to knowing the state, and so is sufficient for knowing $\state^{\View}_{i}$ as well. Therefore, 
\begin{equation}
\sum_j \codet{\state^{\View}_{i}}{\state^{\Viewb}_{j}}=1. \label{eq: Rsqr is row stochastic}
\end{equation}
Hence, the higher is $\codet{\state^{\View}_{i}}{\state^{\Viewb}_{j}}$, the lower must be $\sum_{m\neq j}\codet{\state^{\View}_{i}}{\state^{\Viewb}_{m}}.$

The second constraint is about the total amount of information the issues in $\View$ provide about a single issue from $\Viewb$. This constraint is a mirror image of the previous one: because knowing the variables $(\state^{\View}_m)_{m=1}^{\bign}$ is sufficient for knowing the state, knowing these variables is also sufficient for determining $\state^{\Viewb}_{j}$. Therefore, 
\begin{equation}
\sum_i \codet{\state^{\View}_{i}}{\state^{\Viewb}_{j}}=1.\label{eq: Rsqr is column stochastic}
\end{equation}
An implication of this fact is that the less informative a particular issue $\view_i$ is about $\viewb_j$, the more informative the other issues in $\View$ must be about $\viewb_j$. 

A pertinent consequence of the above constraints is that the values R obtains from learning each issue in $\View$---represented by the vector $\issur=(\issur_1,\ldots,\issur_\bign)^\top$--- must be ``less spread out'' than $\reig = (\reig_1,\ldots,\reig_\bign)^\top$ in the sense of \emph{majorization}. Formally, given two vectors, $\vect, \vectb \in \mathbb{R}^{\bign}$, we say $\vect$ \textbf{majorizes} $\vectb$ if $\vectb = \amat \vect$ for a bistochastic $\bign\times\bign$ matrix $\amat$. In particular, this means that every $\vectb_i$ is a weighted average of elements of $\vect$. 

To connect this definition back to the information constraint, 
define the \textbf{similarity matrix} $\simmat{\View}{\Viewb}$ of $\View$ and $\Viewb$ to be the matrix whose $ij$-entry is the coefficient of determination between $\state^{\View}_i$ and $\state^{\Viewb}_j$---that is, $\codet{\state^{\View}_i}{\state^{\Viewb}_j}$. By the information constraints, \eqref{eq: Rsqr is row stochastic} and \eqref{eq: Rsqr is column stochastic}, this matrix is bistochastic. Since $\issur=\simmat{\View}{\Viewb}\reig$ (see equation~\eqref{eq: R issue payoffs}), it follows the vector $\reig$ majorizes $\issur$. 

There are at least two other equivalent definitions of majorization \citep[see][]{marshall2011inequalities}. The first definition says that $\vect$ majorizes $\vectb$ if for every $m\in\{1,\ldots,\bign\}$ the sum of the $m$ largest entries of $\vect$ is larger than the sum of the $m$ largest entries of $\vectb$, and the sum of all $\bign$ entries of both vectors is the same. The second definition is the mirror image of the first: in addition to all $\bign$ entries of both vectors having an equal sum, the sum of the $m$ smallest entries of $\vect$ must be smaller than the sum of the smallest $m$ entries of $\vectb$ for every $m\in\{1,\ldots,\bign\}$. 

The fact that $\reig$ majorizes $\issur$ is the reason \autoref{prop: eqlbm payoffs}'s part~\eqref{prop: receiver payoff set, payoff part} implies part~\eqref{prop: receiver payoff set, defining inequalities}. For an explanation, note that since $\reig$ majorizes $\issur$, the sum of $\issur$'s least $\litn$ entries and the sum of $\issur$'s largest $\litn$ entries must be less extreme than the corresponding sums for $\reig$. In addition, the average of $\issur$'s entries coincides with the average of $\reig$'s, and must also be between the average of $\issur$'s largest and smallest $\litn$ entries. This latter observation also delivers that one can think of of $\lowur$ and $\highur$ as a splitting of $\medurb$. The fact that the weights associated with this splitting cannot be too extreme follows from noting that $\medurb$ is $\litn$ times the average of $\bign$ numbers, with the lowest $\litn$ summing up to $\lowur$, and the highest $\litn$ summing to $\highur$. 

To show the converse---namely that \autoref{prop: eqlbm payoffs}'s part~\eqref{prop: receiver payoff set, defining inequalities} implies part~\eqref{prop: receiver payoff set, payoff part}---we 
use condition~\eqref{prop: receiver payoff set, defining inequalities} to find some $\issur$ that is majorized by $\reig$ whose highest and lowest $\litn$ entries sum to $\highur$ and $\lowur$, respectively. We then construct a perspective $\View$ such that $\simmat{\View}{\Viewb}\reig=\issur$. A key difficulty in this construction is the fact that not every bistochastic matrix can arise as a similarity matrix for some S-perspective. To circumvent this issue, we rely on the Schur-Horn theorem, which uses majorization to characterize the vectors that can arise as the diagonal of a matrix with a fixed set of eigenvalues.

Proposition~\ref{prop: receiver payoff set} implies that when $\litn=\bign/2$, the information constraints above create a stark tradeoff: the higher the payoff $\slo$ gives R in R's favorite equilibrium, the lower must be the payoff $\slo$ gives R in R's worst equilibrium. This tradeoff is immediately apparent from the Proposition's part~\eqref{prop: receiver payoff set, defining inequalities}. When $\litn=\bign/2$, this part says that $\medurb$ must be the unweighted average of $\lowur$ and $\highur$, and so a higher $\highur$ must be accompanied by a lower $\lowur$. 

The intuition  is easiest to see when $\bign=2$ and $\litn=1$. Suppose $\slo$ and $\rlo$ have unique perspectives $\View$ and $\Viewb$, respectively, and that S communicates the issue $\view_1$ in R's favorite equilibrium. To increase R's payoff in his favorite equilibrium, one would need to make $\view_1$ more informative about R's favorite issue---say $\viewb_1$. By the above constraints, this means that $\view_2$ must become less informative about $\viewb_1$ and more informative about the issue R cares about less, namely $\viewb_2$. Consequently, R's least-favorite equilibrium---which has S communicating about $\view_2$---becomes worse for R. 

This tradeoff above remains when $\litn \neq \bign/2$, but is less stark. For example, an S that maximizes R's worst equilibrium payoff must also minimize R's best equilibrium, and vice versa. We document this fact in the corollary below. To state the corollary, let $\View$ and 
$\Viewb$ be an $\slo$-perspective and 
$\rlo$-perspective, respectively, and let $\reig=(\reig_i)_{i=1}^{\bign}$ be the vector of weights associated with $\rlo$-perspective $\Viewb$. 
\begin{corollary}\label{cor: maximizing worst minimizes best} Given $\rlo$-perspective  $\Viewb$, the following are equivalent:
    \begin{enumerate}[(i)]
    \item R's worst equilibrium payoff is maximized, $\badur(\slo,\rlo)=\medurb.$
    \item R's best equilibrium payoff is minimized, $\goodur(\slo,\rlo)=\medurb.$
    \item The vector $\simmat{\View}{\Viewb}\reig$ is constant for any $\slo$-perspective $\View$.
    \end{enumerate}
\end{corollary}
Thus, an equivalence exists between $\slo$ maximizing R's worst equilibrium payoff, S minimizing R's best equilibrium payoff, and all issues in $\slo$'s perspective $\View$ having the same value to R. To get a sense as to why the corollary holds, suppose $\slo$ and $\rlo$ have a unique perspective. In this case, R's lowest equilibrium payoff is given by the sum of $\litn$ lowest values of $\issur=\simmat{\View}{\Viewb}\reig$. Since $\reig$ majorizes $\issur$ by the information constraints, this sum can equal $\medurb$ if and only if all entries of $\issur$ are equal to $\medurb/\litn$, which, in turn, is equivalent to  the sum of the $\litn$ highest entries of $\issur$ equaling $\medurb$---in other words, R's highest equilibrium value must be $\medurb$. 

Another feature of the $\litn=\bign/2$ case that holds generally is that an $\slo$ that shares a perspective with $\rlo$ both maximizes R's best equilibrium payoff and minimizes R's worst equilibrium payoff across all S loss matrices. 
\begin{observation}\label{obs: Shared perspective is best and worst}
    If $\slo$ and $\rlo$ share a perspective, then $\goodur(\slo,\rlo) = \highurb$ and $\badur(\slo,\rlo) = \lowurb$. 
\end{observation}
Intuitively, if $\slo$ shares a perspective with $\rlo$, then $\slo$ can communicate clearly both about the issues R cares about most and the issues R cares about least. The former is helpful for maximizing R's payoff in his favorite equilibrium. The above-mentioned information constraints mean that the latter is helpful for minimizing R's payoff in his least favorite equilibrium: the more information an S-issue provides about the R-issues R cares about less, the less information that issue must provide about the R-issues R cares about more.

It turns out the above intuition holds more generally. In particular, there is a sense in which the better synchronized S's perspective is with R's, the better she is for R in R's favorite equilibrium, and the worse she is in R's worst equilibrium. The converse also holds: if, regardless of R's priorities, one S perspective is always better (worse) than another in R's favorite (worst) equilibrium, then the former perspective must be better synchronized with R's perspective than the latter one.

This relationship is the content of our next result. To state this result, define a \textbf{$\litn$-average row} of a matrix $\somemat\in\matnn$ to be an unweighted average of $\litn$ distinct rows of $\somemat$, and a \textbf{mixed $\litn$-average row} of $\somemat$ to be a convex combination of $\litn$-averages.\footnote{Given the Birkhoff-von Neumann theorem, a vector is a mixed $\litn$-average row of $\somemat$ if and only if it is a convex combination of the rows of $\somemat$, with each row having weight at most $\tfrac1\litn$. 

}

\begin{proposition}\label{prop: ranking sender perspectives}
Fix a receiver perspective $\Viewb$. Consider senders $\slo_1$ and $\slo_2$ each with a unique perspective $\View_1$ and $\View_2$, respectively. The following are equivalent. 
\begin{enumerate}[(i)]
    \item Any $\rlo$ with perspective $\Viewb$ has $\goodur(\slo_1,\rlo)\geq\goodur(\slo_2,\rlo)$. \label{prop: ranking sender perspectives - best}
    \item Any $\rlo$ with perspective $\Viewb$ has $\badur(\slo_1,\rlo)\leq\badur(\slo_2,\rlo)$. \label{prop: ranking sender perspectives - worst}
    \item Every $\litn$-average row of $\simmat{\View_2}{\Viewb}$ is a mixed $\litn$-average row of $\simmat{\View_1}{\Viewb}$. \label{prop: ranking sender perspectives - averages}
\end{enumerate}
\end{proposition}
Proposition~\ref{prop: ranking sender perspectives} suggests a way of ordering the degree to which two S-perspectives $\View_1$ and $\View_2$ are synced with a given R-perspective $\Viewb$: say \textbf{$\View_1$ more $\Viewb$-synced than $\View_2$} (denoted by $\View_1 \succsim_{\Viewb} \View_2$) if every $\litn$-average row of $\simmat{\View_2}{\Viewb}$ is a mixed $\litn$-average row of $\simmat{\View_1}{\Viewb}$ for all $\litn \in \{1,\ldots,\bign-1\}$. By the above proposition, whenever $\View_1$ is more $\Viewb$-synced than $\View_2$, a receiver with perspective $\Viewb$ who gets to choose his equilibrium prefers an S with a unique perspective of $\View_1$ over an S whose unique perspective is $\View_2$ for \emph{any} $\litn$. Moreover, this ranking is reversed if R always expects to play his worst equilibrium: in that case, R would prefer the $\View_2$-sender over the $\View_1$- sender, regardless of $\litn$. 

The synchronization order is related to a known order in majorization theory. A matrix $\somemat \in \real^{\bign\times\bign}$ \textbf{directionally majorizes} a matrix $\somematt \in \real^{\bign\times\bign}$ if $\somemat \vect$ majorizes $\somematt \vect$ for any vector $\vect \in \real^{\bign}$ \citep[p. 618,]{marshall2011inequalities}. Applying Theorem 3.12 from \cite{peria2005weak} delivers that $\View_1 \succsim_\Viewb \View_2$ holds if and only if $\codet{\View_1}{\Viewb}$ directionally majorizes $\codet{\View_2}{\Viewb}$. 

As one might expect, the perspective $\Viewb$ is more $\Viewb$-synced than any other perspective, since by Proposition~\ref{prop: receiver payoff set} and Observation~\ref{obs: Shared perspective is best and worst}, an S with such a perspective always maximizes R's best equilibrium payoff and minimizes R's worst equilibrium payoff. In other words, the synchonization order always admits a maximum. This order often also admits a minimum: a perspective $\wViewb$ satisfying $\codet{\state^{\wViewb}_i}{\state^{\Viewb}_j}=1/\bign$ for all $i,j$ is less $\Viewb$-synced than any other perspective.\footnote{Such a perspective exists if some order-$\bign$ Hadamard matrix does, which occurs whenever $\bign$ equals any power of 2 \citep[see Theorem 18.3,][]{van2001course}.} Indeed, by \autoref{prop: eqlbm payoffs}, an equilibrium that communicates such a perspective gives R a payoff of $\sum_{i=1}^{\litn}\sum_{j=1}^{\bign}\reig_j \codet{\state^{\wViewb}_i}{\state^{\Viewb}_j} = \tfrac{\litn}{\bign}\sum_{j=1}^{\bign}\reig_j=\medurb$, which by \autoref{prop: receiver payoff set} both minimizes R's best equilibrium utility and maximizes R's worst equilibrium utility.

It is worth noting that the tradeoff between R's payoff under favorable selection and adversarial selection is less stark when R's priorities $\reig$ are known and $\litn \neq \bign/2$. We illustrate this fact below. 

\begin{example}\label{ex: CEO-manager 3d}
Consider the setting of \autoref{ex: CEO-manager}, but assume now the division manager (S) has three products under her control. Thus, we assume now that $\bign=3$. Moreover, both the division manager and the CEO (R) assign the highest priority to providing the division with the appropriate budget for this third product. Specifically, the loss matrices are given by
    \[
    \slo =\threemat{3}{1}{0}{1}{3}{0}{0}{0}{1000},\ \ \ \ \ \  \rlo = \threemat{98}{0}{0}{0}{2}{0}{0}{0}{1000}.
    \]
    We still assume that the two can only talk about one issue (i.e., $\litn=1$), and that states are uncorrelated and have unit variance $\mvar=I$. Under these assumptions, S's perspective and associated weights are
    \[
    \View = \left[\threevec{\tfrac{1}{\sqrt{2}}}{\tfrac{1}{\sqrt{2}}}{0} \threevec{\tfrac{1}{\sqrt{2}}}{-\tfrac{1}{\sqrt{2}}}{0}\threevec{0}{0}{1} \right], 
    \quad 
    \threevec{\seig_1}{\seig_2}{\seig_3}=\threevec{4}{2}{1000},
    \]
    whereas R's perspective and associated weights are
    \[
    \Viewb=\left[\threevec{1}{0}{0} \threevec{0}{1}{0} \threevec{0}{0}{1} \right], 
    \quad
    \threevec{\reig_1}{\reig_2}{\reig_3}=\threevec{98}{2}{1000}.
    \]
    Calculating the similarity matrix $\codet{\View}{\Viewb}$ and issue-value vector $\issur$ gives
    \[
    \codet{\View}{\Viewb} = \threemat{0.5}{0.5}{0}{0.5}{0.5}{0}{0}{0}{1}, \quad \text{and} \quad \issur = \threevec{50}{50}{1000}.
    \]
    Therefore, R receives an expected payoff of $50$ in both the equilibrium that communicates the issue $\view_1$ and the equilibrium that communicates the issue $\view_2$. The equilibrium that communicates the issue $\view_3$ gives R a payoff of $1000$. By Proposition~\ref{prop: eqlbm payoffs}, this payoff maximizes R's utility across all equilibria and across all $\slo$. 
    
    As noted in Observation \ref{obs: Shared perspective is best and worst}, R could also get an equilibrium payoff of $1000$ if he communicated with a sender with a perspective $\Viewb$. Indeed, in this case the similarity matrix would be the identity $\simmat{\Viewb}{\Viewb} = I$, and so the issue-value vector would be $\reig$, meaning R's highest equilibrium payoff is $1000$ as well. 

    Note however, that R's worst equilibrium payoff is lower with an S that shares his perspective $\Viewb$ than it is with an S whose perspective is $\View$: whereas the lowest coordinate of $\reig$ equals $2$, the lowest coordinate of $\issur$ equals $50$. Therefore, while R would be indifferent between a S with perspective $\Viewb$ and an S with perspective $\View$ under favorable equilibrium selection, R would strictly prefer the latter if equilibrium selection were adverserial. 
    \qed 
\end{example}

As demonstrated by the above example, the tradeoff between R's value in his favorite equilibrium and in his worst equilibrium is less stark when $\litn \neq \bign/2$ and $\reig$ is fixed. In particular, some senders are worse for R than others under both equilibrium selection criteria. In fact, when R's payoffs are generic, one can find an S that maximizes R's favorite equilibrium payoff without minimizing R's payoff under adversarial selection. We summarize this observation in the following corollary.

\begin{corollary}\label{cor: payoff set dominance}
For generic $\rlo$, the following are equivalent. \begin{enumerate}[(i)]
    \item\label{cor: payoff set dominance, item same best} Some $\slo$ has $\goodur(\slo,\rlo)=\goodur(\rlo,\rlo)$ and $\badur(\slo,\rlo)>\badur(\rlo,\rlo)$.
    \item\label{cor: payoff set dominance, item strict Pareto} Some $\slo_1,\slo_2$ have $\goodur(\slo_1,\rlo)>\goodur(\slo_2,\rlo)$ and $\badur(\slo_1,\rlo)>\badur(\slo_2,\rlo)$.
    \item\label{cor: payoff set dominance, item litn not half bign} We have $\litn\neq\tfrac\bign2$.\end{enumerate}
\end{corollary}

Whereas the proof relies on \autoref{prop: receiver payoff set}, the corollary's intuition is straightforward. When $\litn < \bign/2$, one can find $\slo$-perspectives that enable S to clearly communicate about the $\litn$ issues R cares the most about, but mix together issues R assigns intermediate importance to with the issues R cares least about.\footnote{When $\litn > \bign/2$, one can analogously mix together R's $\litn$ most important issues, leaving the value of communicating all of them unchanged, but improving R's value from his worst equilibrium.} Such a mixing serves as a hedge, preventing scenarios in which S communicates only information that R views as minutiae. Consequently, some senders can be replaced by senders who improve R's payoff under both equilibrium selection rules---as long as R is generic. For a non-generic set of R, different issues may have the same value (i.e., $\reig_i=\reig_j$ for some $i\neq j$), limiting the ability to hedge against adversarial equilibrium selection via issue mixing.

\section{How much do you want to talk?}\label{sec: investment}

In this section, we assess the value of communication resources. Specifically, we consider a richer game in which our main model is preceded by R making an investment decision that determines $\litn$. If R describes the decision-making protocols in an organization, then the capacity $\litn$ could be determined by organizational design choices, representing investments in adaptive decision-making. Taking our model more literally, we can think of a choice of $\litn$ as a costly attention choice. 

The section's results show that endogenizing $\litn$ means R no longer needs to balance his value in his best and worst equilibrium when choosing between senders. Under adversarial equilibrium selection, R either pays full attention to S (setting $\litn=\bign$), or none (i.e., $\litn=0$), and R's attention decision is the same regardless of S's perspective. By contrast, under favorable selection R's attention decision can take interior values, and S's identity matters: if R expects to hear the S-issues he cares about most, then R always prefers a more synchronized S. Thus, when $\litn$ is endogenous, R prefers more synchronized senders under favorable selection, but does not care about S's identity if he expects equilibrium selection to be adversarial.

Greater precision is in order. Suppose the game starts with R making a public choice of $\litn\in\{0,\cdots,\bign\}$, and then play proceeds according to our main model. S's payoff is as in our main model, whereas R's payoff is equal to that in our main model net of an investment cost, 
\[
\cost\litn.
\]
We assume sequential rationality, in that an equilibrium is played following the investment decision. This equilibrium is determined according to one of two possible selection criteria: either R's best equilibrium is played for every choice of $\litn$, or R's worst equilibrium is played.

We begin by analyzing R's decision under favorable selection. Suppose for exposition that S and R are generic, and so have unique perspectives, $\View$ and $\Viewb$, and let $\reig$ be R's associated weights. Suppose further that the issues in $\View$ are ordered in terms of their (generically unique) value to the receiver; that is, $\issur:=\codet{\View}{\Viewb}\reig$ is such that $\issur_1>\cdots>\issur_\bign$. Given this ordering, choosing a capacity of $\litn$ results in S communicating issues $1,\ldots,\litn$ in R's favorite equilibrium. Therefore, R's payoff under favorable selection from choosing $\litn$ is
\[
\goodur(\slo,\rlo) - \cost \litn =  \sum_{i=1}^{\litn} \issur_i - \cost\litn = \sum_{i=1}^{\litn}(\issur_i - \cost).
\]
Clearly, R strictly benefits (on the margin) from listening to issue $\litn$ if and only if  $\issur_\litn > \cost$. Moreover, capacity exhibits decreasing marginal returns, meaning R benefits from listening to the $\litn^{\text{th}}$ issue only if he also benefits from listening to the $i^{\text{th}}$ issue for any $i < \litn$. Therefore, the optimal capacity under favorable selection is determined by a first-order condition: $\issur_{\litn}\geq \cost \geq \issur_{\litn+1}$. 

What if R expects adversarial selection? In this case, listening to $\litn$ issues gives R a payoff of 
\[
\badur(\slo,\rlo) - \cost\litn = \sum_{i=\bign-\litn+1}^{\bign}(\issur_i - \cost).
\]
The above objective exhibits \emph{increasing} marginal returns to capacity. Hence, if listening to $\litn$ issues is better than not listening at all, then listening to $\litn+1$ issues is even better. So R's optimal capacity under adversarial selection is bang-bang: R either gives S his full attention, setting $\litn=\bign$, or does not listen at all. Which of these two options is optimal depends on the costs. Specifically, choosing a capacity of $\bign$ is optimal if and only if 
\[
\frac{1}{\bign}\sum_{i=1}^{\bign}\issur_i \geq \cost.
\]
Otherwise, R does not listen to S and so chooses zero capacity. 

Thus, whereas R's chosen capacity under favorable selection is determined by the value of the best \emph{marginal} issue, his capacity under adversarial selection depends on the value of the \emph{average} issue. Consequently, R's capacity choice differs between these two equilibrium selection regimes. \autoref{fig: investment} illustrates these differences, denoting R's optimal investment under best-case selection by $\litnb$, and under worst-case selection by $\litnb$. 

\begin{figure}
    \centering
    \includegraphics[width=0.5\linewidth]{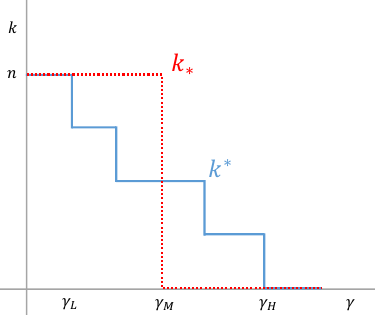}
    \caption{Optimal investment $\litn$ by R under best-case and worst-case equilibrium, respectively, as a function of the cost parameter $\cost$.}
    \label{fig: investment}
\end{figure}

When $\cost$ is sufficiently low, R sets $\litnb=\litnw=\bign$ and pays full attention to S in both regimes. For $\cost$ slightly above the value of R's least favorite issue $\issur_\bign$, R reduces his attention under favorable selection $\litnb<\bign$, but keeps paying full attention under adversarial selection $(\litnw=\bign)$. As $\cost$ increases further, R's attention under favorable selection decreases gradually, going to zero once $\cost$ passes the value $\issur_1$ of R's favorite S-issue. By contrast, an R expecting the worst equilibrium keeps giving S his full attention until $\cost$ passes the value of the average issue $\frac{1}{\bign}\sum_{i=1}^{\bign}\issur_i$, at which stage R's attention drops from $\litnw=\bign$ to $\litnw=0$. \autoref{prop: investment} summarizes these differences.

\begin{proposition}\label{prop: investment}
Some cutoffs $\hcost\geq\mcost\geq\lcost>0$ exist (generically, $\hcost>\mcost>\lcost$) such that any R-best-case and R-worst-case optimal investments $\optlitnb$ and $\optlitnw$ satisfy: \begin{enumerate}[(i)]
    \item If $\lcost<\cost<\mcost$, then $\optlitnb<\optlitnw=\bign$;
    \item If $\mcost<\cost<\hcost$, then $\optlitnb>\optlitnw=0$;
    \item If $\cost<\lcost$ or $\cost>\hcost$, then $\optlitnb=\optlitnw\in\{0,\bign\}$.
\end{enumerate}
\end{proposition}

R's ability to change her capacity also changes the way he ranks different senders. If R expects the worst equilibrium, he invests the same amount regardless of S. The reason is that the value of the average issue does not depend on S's perspective:
\[
\frac{1}{\bign}\sum_{i=1}^{\bign}\issur_i 
= \frac{1}{\bign}\sum_{i=1}^{\bign}\sum_{j=1}^{\bign}\codet{\state^{\View}_{i}}{\state^{\Viewb}_{j}}\reig_j 
= \frac{1}{\bign}\sum_{j=1}^{\bign}\reig_j\sum_{i=1}^{\bign}\codet{\state^{\View}_{i}}{\state^{\Viewb}_{j}}
= \frac{1}{\bign}\sum_{j=1}^{\bign}\reig_j.
\]

By contrast, R cares a lot about which S he faces under favorable selection, a fact which we now demonstrate.
\addtocounter{example}{-2}
\begin{example}[continued] \label{ex: CEO-manager with investment}
    Let us return to the setting of \autoref{ex: CEO-manager} where $\bign=2$, the variance-covariance matrix is the identity $\mvar=I$, and the loss matrices are
    \[
    \slo =\twomat{3}{1}{1}{3} \text{ and }\rlo = \twomat{98}002.
    \]
    S and R's corresponding perspectives and associated weights  are given by
    \[
    \View = \tfrac{1}{\sqrt{2}} \left[\twovec{1}{1} \twovec{1}{-1} \right], (\seig_1,\seig_2)=(4,2) \quad \text{and} \quad \Viewb=\left[\twovec{1}{0} \twovec{0}{1} \right], (\reig_1,\reig_2)=(98,2).
    \]
    Suppose further that R publicly chooses the value of $\litn$ at a cost of $25\litn$ before the game. Since both $\view_1$ and $\view_2$ gives R a payoff of $50$, R chooses to pay full attention to S under both favorable and adversarial equilibrium selection. Thus, we get $\litnb = \litnw=2$, and R's total payoff is $50$. 

    What if S's loss is given by $\slob=R$ instead? Under adversarial selection, R still chooses full attention and gets a payoff of $50$. The same is not true under favorable selection. For an explanation, note that if $\litn=1$, S can either communicate $\state_1$ or $\state_2$ to R, giving R a payoff of $98$ under the former, and a payoff of $2$ under the latter. Therefore, an R who expects his best equilibrium chooses to hear only the value of $\state_1$, getting a net payoff of $98-25=73$. Note this payoff is strictly higher than under $\slo$. \qed
\end{example}

The example compares an S who is minimally synced with R with an S who is as synced with R as possible, and shows R prefers the latter over the former. This ranking, however, is not confined to the extremes of the synchronization order: the more synced S is with R, the better off R is under favorable selection. The reason is that a more synced S is better for R under favorable selection for every $\litn$, and so is clearly also better under the optimal $\litn$. 

Notice though that having a more synced S does not mean that R pays S greater attention. The example is a case in point: R chooses a smaller $\litnb$ when faced with the maximally synced S than he chooses when faced with a minimally synced one. Our next result shows this comparative static  holds more generally, applying to other pairs of senders who are are strictly ranked via the synchronization order.

\begin{proposition}\label{lem: optimal investment switching}
Suppose $\rlo$ has perspective $\Viewb$, 
and let $\hslo$ and $\lslo$ be two sender losses whose unique perspectives are $\hView$ and $\lView$, respectively. If $\hView\succsim_\Viewb\lView$ and $\goodur(\hslo,\rlo) \neq \goodur(\lslo,\rlo)$ for some $\litn$, then some $\hhcost>\hcost\geq\lcost>\llcost>0$ exist such that the set of R-best-case optimal investments $\hoptlitnb$ and $\loptlitnb$ for $\hslo$ and $\lslo$, respectively, satisfy:\footnote{We order $\hoptlitnb$ and $\loptlitnb$ via the strong set order.}
\begin{enumerate}[(i)]
    \item If $\llcost<\cost<\lcost$, then $\hoptlitnb<\loptlitnb$;
    \item If $\hcost<\cost<\hhcost$, then $\hoptlitnb>\loptlitnb$;
    \item If $\cost<\llcost$ or $\cost>\hhcost$, then $\hoptlitnb=\loptlitnb$.
\end{enumerate}
Consequently, the above holds for generic $\rlo$-weights $\reig$ whenever $\hView\succ_\Viewb\lView$.
\end{proposition}

Thus, under favorable selection, a generic R listens more to a more synchronized S when attention is moderately expensive, and listens less when attention is somewhat cheap. Intuitively, the more synchronized an S is with R, the clearer S communicates about the issues R cares least and most about. Consequently, the more synchronized S is with R, the value of the best issue S can communicate increases, while value of the worst issue S can communicate decreases. Therefore, when costs are high, making S more synchronized might result in R switching from not listening to S all, to agreeing to listen to the S-issues R values most. Similarly, when costs are low, making S more synchronized could push R to reduce his attention by dropping the S-issues that are less important to him.

\section{What do you want us to discuss?}\label{sec: disc}

We conclude by discussing the implications of some of our modeling assumptions. 

\paragraph{Partial Rank.} The number $\litn$ parametrizes the capacity of R's attention. In our main analysis, we assume  R utilizes all of this capacity, requiring R's strategy to have full rank. 
This assumption represents a scenario in which R's attention is free at the margin at the time he communicates with S. 
An alternative model would allow R to choose an arbitrary linear decision rule $\rstrat\in\matnk$ that need not have full rank; hence capacity up to $\litn$ is free at the margin, but R may abstain from using it. 

Given our analysis, one can characterize the equilibrium outcomes of this alternative game. For each $\litnt\in\{0,\ldots,\litn\}$ and each $\slo$-perspective $\View$, we can consider a strategy profile in which S sends message $\mess=(\state_i^\View\mathbf1_{i\leq\litnt})_{i=1}^\bign$
and R chooses action $\act$ such that $\act_i^\View=\mess_i\mathbf1_{i\leq\litnt}$. In words, S fully communicates the first $\litnt$ issues according to her perspective, and R takes these at face value and ignores the remaining dimensions of the message. Our analysis extends to show all of these strategy profiles are equilibria, and every equilibrium has an outcome-equivalent one from this family.\footnote{In particular, all equilibria of the original game remain equilibria even though R has new deviation opportunities.} 

Note that exhausting the capacity on communicating this perspective is a strong Pareto improvement, so every equilibrium with $\litnt<\litn$ constitutes a coordination failure---the most extreme one having $\litnt=0$ so that action zero is always chosen.

These coordination failures, however, rely on a fragile kind of self-fulfilling expectations: R ignores some dimensions because he is certain that they are not helpful for listening to S, while S does not use these dimensions because she expects R to ignore them. These expectations are quite fragile, however, as they rely on each player being absolutely certain that the other is wasting some capacity. For an illustration, suppose R has a positive probability of being restricted to full-rank strategies, and S sends her message in ignorance of this realization. Then one can show that every equilibrium has R using a full-rank decision rule even when he is not forced to (and using the same decision rule in either case). Therefore, one can think of our full-rank restriction to studying equilibria that are robust to the possibility of non-wastefulness.

\paragraph{Uncenteredness} We have assumed to this point that the state is centered, that is, it has zero expectation. Suppose now that the state may have a nonzero expectation $
\defstate$. The comparison between this model and our original model depends on what we assume about R's space of available strategies.

First, suppose R's space of strategies is the same as in our original model, meaning R must choose a full-rank linear strategy. In this strategy space the zero action occupies a special role: No matter what strategy R chooses, S always has some message to send (the zero message) that ensures the zero action is played. In our main model this fact---combined with linearity of S's best reply---results in R's action being an unbiased estimate of the state. The same is not true in the uncentered model, since the range of R's strategy might not contain $\defstate$. In this case, 
R's decision becomes a biased estimate of the state. This bias obviously has welfare consequences, but it also shapes R's incentives. Whereas in our main model R acts to minimize his exposure to (weighted) variance of the state, in this model he balances this goal against minimizing his exposure to the bias. In fact, a straightforward modification of \autoref{sec: proofs}'s \autoref{lem: trace payoff calculation} proof tells us this model is identical---with the same set of equilibria and resulting expected payoffs---to that with a modified covariance matrix of $\mvar+\defstate\defstate^\top$. Hence, imposing a default action different from the default state is equivalent to R being more uncertain about the state in the direction of the imposed bias in decision making.

An alternative model is one in which R chooses a pair $(\defact,\rstrat)$, where $\defact\in\real^\bign$ is a default action and $\rstrat\in\matnk$ has full rank. This model is essentially equivalent to our general model. Indeed, R incentives require that the expected action be equal to the expected state. As a consequence, equilibrium communication and play is identical to our original model with a modified state $\stateb=\state-\defstate$. Up to outcome equivalence, equilibria take the following form: We decompose the centered state $\stateb$ according to some S perspective, R chooses a default action $\defstate$, and the players communicate this perspective concerning $\stateb$. So if R can freely choose his default action, then our centeredness assumption is purely a normalization.

\paragraph{Bias}
Our model assumes the players have the same ideal action given the state. This assumption allows us to focus on the incentive frictions that stem from our bounded rationality assumption. Nevertheless, one might wonder how such bounded rationality interacts with a misalignment between S's and R's ideal action. Specifically, suppose S's ideal action in state $\state$ is $\state+\bias$ for some nonzero vector $\bias\in\real^\bign$. 

Suppose that, as discussed in the case of an uncentered state, R may choose a default action in addition to a linear decision rule. Hence, R's strategy is a pair $(\defact,\rstrat)$. S's incentives now require the action to be a rank-$\litn$ linear function of her ideal action $\state+\bias$. Meanwhile, R's incentives require that the average action be zero (which is the average state). Hence, any equilibrium admits an outcome-equivalent equilibrium in which the default action is zero. Consequently, the conditions for equilibrium in this modified model are the same as in our main model, but with the added restriction that R's resulting action is simultaneously both a linear function of $\state$ and a linear function of $\state+\bias$. In the language of \autoref{prop: eqm char}, all equilibria are (up to outcome equivalence) given by strategy profiles that communicate an $\slo$-perspective $\View$ such that every $i\in\{1,\ldots,\litn\}$ has $\bias^\View_i=0$. We can interpret this orthogonality condition as saying S cannot be biased about any communicated issue. Unfortunately, these equilibrium conditions are impossible to satisfy for most $\slo$ matrices.

\paragraph{Commitment}

We have assumed the players' play to be strategically simultaneous: R chooses a decision rule without observing how S communicates, and S (knowing the state) decides how to communicate without observing R's decision rule. Other natural communication protocols would allow one of the players to commit to their behavior.

First, suppose S could commit to a strategy, and R would respond with his choice of a full-rank $\bign\times\litn$  matrix. Such a setting would amount to a Bayesian persuasion model \citep{Kamenica2011} in which R is subject to our imposed bounded rationality constraint. The next result establishes that the solution to this persuasion problem is in fact an equilibrium. Hence, the sole benefit of commitment power to S is the ability to select an equilibrium.

\begin{proposition}\label{prop: persuasion}
A strategy profile is S-preferred if and only if it is outcome equivalent to an S-preferred equilibrium.
\end{proposition}

Note the above result says that neither player's incentive constraints in our game impose costs on S. If S could choose any strategy profile at all in our game, she would effectively be choosing among all state-contingent action choices, subject to the constraint that all chosen actions live in some $\litn$-dimensional subspace of $\real^\bign$. The proposition shows S's preferred strategy profile simply communicates the $\litn$ most important issues from her perspective, which is an equilibrium of our model. We remark that Corollary~1 from \cite{hancart2025simple} essentially yields a specialization of this result for case in which $\slo = \rlo$, there are only $\bign=2$ dimensions, and R's capacity is $\litn=1$.

The situation is quite different if R can commit. If R first publicly commits to a strategy, the interaction is one of \emph{subspace delegation}: R announces a $\litn$-dimensional set of actions from which S can freely choose. This delegation model can be seen as an organizational design problem. Maintaining too much flexibility in an organization is infeasible, but an organization can potentially decide which kinds of flexibility to invest in.\footnote{This problem is similar to \cite{frankel2016delegating} in that it studies a variant of \cite{holmstrom1980theory} with a multidimensional state. Relative to \cite{frankel2016delegating}, our subspace delegation problem replaces an additive bias with a discrepancy in priorities and puts further structure on admissible delegation sets.}

The next result shows that R benefits from this commitment power in most specifications of our model. In particular, S commitment and R commitment have very different implications for issue communication.

\begin{proposition}\label{prop: delegation}
R strictly benefits from commitment for generic $(\slo,\rlo)$.
\end{proposition}

For intuition, let us return to CEO and manager of \autoref{ex: CEO-manager}. Recall that the CEO (R) chiefly wants to choose the correct budget for division $1$, whereas the manager (S) mainly wants to set the correct total budget. One equilibrium communicates the ideal total budget, one communicates the ideal difference between the two divisions' budgets, and both give a payoff of $50$ to the CEO. But what if the CEO could commit to only take the manager's advice on division 1, and to stick with the ex-ante optimal budget for the other division---i.e., $\rstrat=[1\ 0]$? A direct computation shows that the manager would respond by sending message $\mess=\state_1+\tfrac13\state_2$. She inflates the requested budget for division 1 when the ideal budget of division 2 is high, because better targeting the total budget outweighs the added distortion on relative budgets. Meanwhile, this strategy profile gives the CEO an expected payoff of $\tfrac{784}{9}\approx87.1$.

The CEO is able to attain the above payoff substantially higher than his equilibrium payoff of $50$ only because he is using his commitment power. For an explanation, note that the CEO's best linear estimate of $\state_2$ when hearing a message of $\mess$ is a strictly positive multiple of $\mess$, and that his best estimate of $\state_1$ deflates $\mess$ toward zero, because $\mess$ equals $\state_1$ plus independent ``noise''. Hence, the CEO would like to deviate in both coordinates of $\rstrat$. Committing not do so enables him to benefit from superior advice.

\bibliography{ref-linearmachines}

\appendix
\newpage

\section{Proofs for Main Results}\label{sec: proofs}

\begin{definition}
For any positive definite matrix $\somemat\in\matnn$:\begin{enumerate}
    \item Let $\ip\cdot\cdot_\somemat$ denote the inner product $\real^\bign$ given by $\ip{\somevec}{\somevecb}:=\somevec^\top\somemat\somevecb$.
    \item Say $\strats\in\matnn$ is an \textbf{$\somemat$-projection} if any of these equivalent conditions hold: \begin{itemize}
        \item $\strats$ is an orthogonal projection on the inner product space $\paren{\real^\bign,
        \ip\cdot\cdot_\somemat}$.
        \item $\strats$ is idempotent (i.e.,  $\strats^2=\strats$) with $\strats^\top\somemat=\somemat\strats$.
        \item $\sqrt{\somemat}\strats\inv{\sqrt{\somemat}}$ is idempotent and symmetric.
        \item $\sqrt{\somemat}\strats\inv{\sqrt{\somemat}}$  is an orthogonal projection on the $\real^\bign$ with the dot product.
        \item $\strats^\top\somemat\strats=\somemat\strats$.
    \end{itemize}
\end{enumerate}
\end{definition}

\begin{notation}
Left multiplication by any matrix $\sstrat\in\matkn$ is an S strategy, which we also denote by $\sstrat$. Conversely, any linear S strategy $\sstrat$ can be represented as left multiplication by a matrix, and we also let $\sstrat\in\matkn$ denote this matrix.
\end{notation}

\begin{lemma}\label{lem: S best response}
Let $(\sstrat,\rstrat)$ be a strategy profile. Then $\sstrat$ is an S best response to $\rstrat$ if and only if $\sstrat$ is linear and $\strats:=\rstrat\sstrat\in\matnn$ is an $\slo$-projection of rank~$\litn$. Moreover, a unique S best response exists for any R strategy.
\end{lemma}
\begin{proof}

Given any $\stater\in\real^\bign$, we have: \begin{eqnarray*}
&& \sstrat(\stater)\in\argmax_{\messr\in\real^\litn}
    \curlyb{
    \stater^\top \slo\stater- \paren{\stater-\rstrat\messr}^\top \slo\paren{\stater-\rstrat\messr}
    } \\
&\iff& 
\rstrat\sstrat(\stater)\in\argmin_{\actr\in\ran\rstrat}
    \curlyb{
    \paren{\stater-\actr}^\top \slo\paren{\stater-\actr}
    } \\
&\iff& 
\rstrat\sstrat(\stater) \text{ is the } \ip\cdot\cdot_\slo\text{-nearest point to } \stater \text{ in } \ran\rstrat \\    
&\iff& 
\rstrat\sstrat(\stater) \text{ is the } \ip\cdot\cdot_\slo\text{-orthogonal projection of } \stater \text{ onto } \ran\rstrat.  
\end{eqnarray*}
Therefore, $\sstrat$ is an S best response to $\rstrat$ if and only if $\sstrat$ is linear and $\rstrat\sstrat(\cdot)$ is the $\ip\cdot\cdot_\slo$-orthogonal projection map onto $\ran\rstrat$. 

Because left multiplication by $\rstrat$ is an injective linear map on $\real^\litn$, it follows that this map has an inverse, which is linear. Therefore, given the above equivalence, S has a unique best response to $\rstrat$, namely, the composition of this inverse with the $\ip\cdot\cdot_\slo$-orthogonal projection of $\stater$ onto $\ran\rstrat$ (a linear map).  

Given $\sstrat\in\matkn$, let $\strats:=\rstrat\sstrat\in\matnn$. We have argued $\sstrat$ is an S best response to $\rstrat$ only if $\strats$ is the $\slo$-projection onto $\ran\rstrat$---or, equivalently, an $\slo$-projection with $\ran\strats=\ran\rstrat$. But because $\ran\strats$ is a linear subspace of $\ran\rstrat$, they coincide if and only if they have the same dimension (and the latter has dimension $\litn$).
\end{proof}

\begin{lemma}\label{lem: trace payoff calculation}
Let $(\sstrat,\rstrat)$ be a strategy profile with $\sstrat$ linear, and let $\strats:=\rstrat\sstrat\in\matnn$. Then S and R, respectively, have expected payoff 
$$
\expec{\us}=\Tr\brac{\paren{2\slo\strats-\strats^\top\slo\strats}\mvar}
\text{ and }
\expec{\ur}=\Tr\brac{\paren{2\rlo\strats-\strats^\top\rlo\strats}\mvar}.
$$
Hence, if $\strats$ is a $\inv{\mvar}$-projection, then $$
\expec{\us}=\Tr\paren{\slo\strats\mvar}
\text{ and }
\expec{\ur}=\Tr\paren{\rlo\strats\mvar}.
$$
\end{lemma}
\begin{proof}
First, any $\somemat\in\matnn$ has 
$$\expec{\state^\top\somemat\state} 
= \expec{\Tr\paren{\somemat\state\state^\top}} 
= \Tr\paren{\somemat\expec{\state\state^\top}} 
= \Tr\paren{\somemat\mvar}.$$
Meanwhile, S's payoff takes the form 
\begin{eqnarray*}
\us&=& \state^\top \slo\state- (\state-\strats\state)^\top \slo(\state-\strats\state) \\
&=& \state^\top\brac{\slo- (I-\strats^\top)\slo(I-\strats)}\state \\
&=& \state^\top\paren{\strats^\top\slo+\slo\strats-\strats^\top\slo\strats}\state \\
&=& \state^\top\paren{2\slo\strats-\strats^\top\slo\strats}\state^\top.
\end{eqnarray*}
The formula for $\expec{\us}$ comes from combining these two calculations. Specializing now to the case that $\strats$ is a $\inv{\mvar}$-projection, we get
$$\Tr\paren{\strats^\top\slo\strats\mvar}
=\Tr\brac{\slo\strats\mvar\paren{\strats^\top\inv{\mvar}}\mvar}
=\Tr\brac{\slo\strats\mvar\paren{\inv{\mvar}\strats}\mvar}
=\Tr\brac{\slo\strats^2\mvar}
=\Tr\brac{\slo\strats\mvar},
$$
yielding the simplified payoff formula. Identical calculations (replacing $\slo$ with $\rlo$) give the payoff formulae for $\expec{\ur}.$
\end{proof}

\begin{lemma}\label{lem: R best response}
Let $(\sstrat,\rstrat)$ be a strategy profile with $\sstrat$ linear, and let $\strats:=\rstrat\sstrat\in\matnn$.
Then $\rstrat$ is an R best response to $\sstrat$ if and only if $\strats$ is a $\inv{\mvar}$-projection. Moreover, a unique R best response exists for any rank-$\litn$ linear S strategy.
\end{lemma}
\begin{proof}
Define the differentiable function $\gain:\matnn\to\real$ by letting $$\gain(\stratsb):=\Tr\brac{\paren{2\rlo\stratsb-\stratsb^\top\rlo\stratsb}\mvar}
=2\Tr(\rlo\stratsb\mvar)-\Tr\paren{\stratsb^\top\rlo\stratsb\mvar}.$$
We begin by noting the following conditions are equivalent for an R strategy $\rstrat$: \begin{enumerate}[(i)]
    \item $\rstrat$ is a best response for R.
    \item $\rstrat\in\argmax_{\rstratb\in\matnk \text{ full rank}} \gain(\rstratb\sstrat).$
    \item $\partialop{\epsilon}\big|_{\varepsilon=0} \gain\paren{\brac{\rstrat+\varepsilon\somemat}\sstrat}=0$ for every $\somemat\in\matnk$.
    \item $\rstrat\in\argmax_{\rstratb\in\matnk} \gain(\rstratb\sstrat).$
\end{enumerate}
First, \autoref{lem: trace payoff calculation} tells us (i) and (ii) are equivalent, and (iv) obviously implies (ii). Next, (ii) implies (iii) because the set of full-rank matrices is open in $\matnk$, and so $\rstrat+\varepsilon\somemat$ has full rank for any $\somemat\in\matnk$ and all $\varepsilon\in\real$ close enough to zero. Then, to see (iii) implies (iv), it suffices to see that $\rstratb\mapsto\gain(\rstratb\sstrat)$ is concave. To that end, note that $\rstratb\mapsto\rstratb\sstrat$ and $\stratsb\mapsto2\Tr(\rlo\stratsb\mvar)$ are linear, and the function  $\stratsb\mapsto\Tr\paren{\stratsb^\top\rlo\stratsb\mvar}$ is (as a quadratic and nonnegative function) convex---nonnegativity follows because any $\stratsb$ has $\Tr\paren{\stratsb^\top\rlo\stratsb\mvar}=\Tr\brac{\paren{\sqrt\rlo\stratsb\rmvar}^\top\paren{\sqrt\rlo\stratsb\rmvar}}$ and a positive semidefinite matrix has nonnegative trace. 
Given that the four conditions are equivalent, we work with condition (iii) below. 

Letting $\strats:=\rstrat\sstrat$, any $\somemat\in\matnk$ has 
\begin{eqnarray*}
\gain\paren{\brac{\rstrat+\varepsilon\somemat}\sstrat}
&=&2\Tr(\rlo\strats\mvar) -\Tr\paren{\strats^\top\rlo\strats} 
- \varepsilon^2\Tr\paren{\sstrat^\top\somemat^\top\rlo \somemat\sstrat\mvar}
\\
&&
+2\varepsilon\Tr\paren{\rlo\somemat\sstrat\mvar}
- 2\varepsilon\Tr\paren{\strats^\top\rlo \somemat\sstrat\mvar}\\
\implies 
\tfrac12\partialop{\epsilon}\big|_{\varepsilon=0} \gain\paren{\brac{\rstrat+\varepsilon\somemat}\sstrat}
&=& \Tr(\rlo\somemat\sstrat\mvar)
- \Tr(\strats^\top\rlo \somemat\sstrat\mvar)\\
&=& \Tr(\sstrat\mvar\rlo\somemat)- \Tr(\sstrat\mvar\strats^\top\rlo \somemat) \\
&=& \Tr\brac{\sstrat\mvar(I-\strats^\top)\rlo\somemat} 
\end{eqnarray*}
So $\rstrat$ is a best response for R if and only if 
$\sstrat\mvar(I-\strats^\top)\rlo\somemat$ has zero trace for every $\somemat\in\matnk$, or equivalently if and only if $\sstrat\mvar(I-\strats^\top)\rlo=0$. Right multiplying this equation by the invertible matrix $\rlo^{-1}$ and left multiplying by the full-column-rank matrix $\rstrat$, we can equivalently express this condition as $\strats\mvar(I-\strats^\top)=0$. Transposing, $\rstrat$ is a best response for R if and only if $\mvar\strats^\top=\strats\mvar\strats^\top$. This condition is equivalent to $\strats^\top$ being a $\mvar$-projection, which in turn is equivalent to $\strats$ being a $\inv{\mvar}$-projection.

Let us next see that any full-rank linear S strategy $\sstrat$ admits an R best response. Given the above argument, we need to find an R strategy $\rstrat$ such that $\rstrat\sstrat$ is a $\inv{\mvar}$-projection. And indeed, by the rank-nullity theorem, left multiplication by $\sstrat$ bijectively maps the $\inv{\mvar}$-orthogonal complement $\vectspace$ of $\ker\sstrat$ to $\ran\sstrat=\real^\litn$, so we can choose $\rstrat$ so that left multiplication by it is the inverse of this map. Then, left multiplication by $\rstrat\sstrat$ fixes $\vectspace$ and annihilates $\ker\sstrat$, making it the $\inv{\mvar}$-projection onto $\vectspace$.

Finally, let us observe that the above R best response is unique. The reason is that any R strategy has full column rank, and so the composition of any R strategy with a given linear S strategy has the same kernel. Uniqueness then follows from an orthogonal projection being determined by its kernel.
\end{proof}

\begin{lemma}\label{lem: perspective representation of A matrices}
For $\strats\in\matnn$, letting  $\normproj:=\brac{\mathbf1_{i=j\leq\litn}}_{i,j=1}^\bign$, the following are equivalent: \begin{enumerate}
    \item The matrix $\strats$ has rank $\litn$ and is both a $\slo$-projection and a $\inv\mvar$-projection.
    \item Some $\slo$-perspective $\View$ exists such that $\strats=\inv{(\View^\top)}\normproj \View^\top$.
\end{enumerate}
\end{lemma}
\begin{proof}
First, suppose $\strats=\inv{(\View^\top)}\normproj \View^\top$ for some $\slo$-perspective $\View=[\view_1\cdots\view_\bign]$. 
Let us now see $\strats$ is a rank-$\litn$ map that is both a $\slo$-projection and a $\inv\mvar$-projection. First, $\View$ (and hence its transpose) having full rank means $\strats$ has rank equal to $\rank\normproj=\litn$. Second, $\strats$ is idempotent because the similar matrix $\normproj$ is. Third, that $\View$ is a perspective means $\View^\top\mvar\View=I$ and so $\strats=\mvar\View\normproj \View^\top$. Therefore, $\inv{\mvar}\strats=\View\normproj \View^\top=\strats^\top\inv{\mvar}$, making $\strats$ a $\inv{\mvar}$-projection. Finally, $\View$ being an $\slo$-perspective means $\slo=\View\somedmat\View^\top$ for some diagonal matrix $\somedmat\in\matnn$. Therefore, $$\strats^\top\slo-\slo\strats=\View\normproj\inv{\View}\View\somedmat\View^\top-\View\somedmat\View^\top\inv{(\View^\top)}\normproj\View^\top=\View\paren{\normproj\somedmat-\somedmat\normproj}\View^\top,$$
which is zero because $\somedmat$ and $\normproj$ are both diagonal. Hence, $\strats$ is an $\slo$-projection.

Conversely, suppose $\strats$ has rank $\litn$ and is both a $\slo$-projection and a $\inv\mvar$-projection. In what follows, define the rank-$\litn$ matrix $\stratsb:=\inv{\rmvar}\strats\rmvar$ and the positive definite matrix $\slob:=\rmvar\slo\rmvar$. 
That $\strats$ is a $\inv\mvar$-projection tells us $\stratsb$ is an orthogonal projection (with respect to the standard inner product); and that $\strats$ is an $\slo$-projection tells us 
$$0=\rmvar(\slo\strats-\strats^\top\slo)\rmvar=\slob\stratsb-\stratsb^\top\slob=\slob\stratsb-\stratsb\slob.$$
The symmetric matrices $\slob$ and $\stratsb$ commute, meaning they are simultaneously diagonalizable by an orthonormal (according to the standard inner product) basis. Since $\stratsb$ is an orthogonal projection matrix, this means there exists an orthonormal (again, according to the standard inner product) basis $(\viewt_i)_{i=1}^\litn$ for $\ran\stratsb$ and an orthonormal basis $(\viewt_i)_{i=\litn+1}^\bign$ for $\ker\stratsb$ such that all of $\viewt_1,\ldots\viewt_\bign$ are eigenvectors for $\slob$. Now, define 
$\Viewt:=[\viewt_1\cdots\viewt_\bign]\in\matnn$ and $\View:=\inv{\rmvar}\Viewt\in\matnn$. Let us see $\View$ witnesses the desired property for $\strats$. First, because $\Viewt$ is orthogonal by construction, $\View$ is a perspective. Next, the orthogonal matrix $\Viewt$ has $\slob$-eigenvectors as its columns, implying $\slob=\Viewt\somedmat\Viewt^\top$ for some diagonal $\somedmat\in\matnn$. Therefore, 
$$\View\somedmat\View^\top
=\paren{\inv{\rmvar}\Viewt}\paren{\Viewt^\top\slob\Viewt}\paren{\Viewt^\top\inv{\rmvar}}
=\inv{\rmvar}\slob\inv{\rmvar}=\slo.
$$
Hence, $\View$ is an $\slo$-perspective. Finally, observe that $\View^\top=\Viewt^\top\inv{\rmvar}$ which has inverse $\rmvar\Viewt$, and that 
$$\Viewt\normproj \Viewt^\top=
\Viewt\paren{\sum_{i=1}^\litn\elem_i\elem_i^\top}\Viewt^\top = \sum_{i=1}^\litn\viewt_i\viewt_i^\top=\stratsb
.$$
Thus, 
$\inv{(\View^\top)}\normproj \View^\top
=\rmvar\Viewt\normproj \Viewt^\top\inv{\rmvar}
=\rmvar\stratsb\inv{\rmvar}=\strats,
$
as required.
\end{proof}

\begin{proof}[Proof of \autoref{prop: eqm char}]
By \autoref{lem: S best response} and \autoref{lem: R best response}, a given strategy profile $(\sstrat,\rstrat)$ is an equilibrium if and only if $\sstrat$ is linear and $\strats:=\rstrat\sstrat\in\matnn$ is a rank-$\litn$ matrix that is both a $\slo$-projection and a $\inv\mvar$-projection.
Meanwhile, for the strategy profile $(\sstrat,\rstrat)$ that communicates perspective $\View=[\view_1\cdots\view_\bign]$, we know $\sstrat\in\matkn$ is a linear strategy, and 
$$
\sstrat=\begin{pmatrix} 
	I_\litn & 0_{\litn,\bign-\litn}\\
	\end{pmatrix}\View^\top
\text{ and }
\rstrat=\inv{(\View^\top)}\twovec{I_\litn}{0_{\bign-\litn,\litn}}.
$$
Letting $\normproj:=\brac{\mathbf1_{i=j\leq\litn}}_{i,j}\in\matnn$ then gives $\rstrat\sstrat=\inv{(\View^\top)}\normproj \View^\top.$
The desired equivalence then follows from \autoref{lem: perspective representation of A matrices}.
\end{proof}

\begin{notation}
For any perspectives $\View$ and $\Viewb$, define their \textbf{correlation} matrix $\cormat{\View}{\Viewb}\in\matnn$, whose $ij$-entry is the correlation $\expec{\state^\View_i\state^\Viewb_j}$. Algebraically, the correlation matrix takes the form $\cormat{\View}{\Viewb}=\View^\top\mvar\Viewb$.
\end{notation}

\begin{lemma}\label{lem: similarity payoff calculation}
Let $\Viewb$ be an $\rlo$-perspective with associated weights $\reig=(\reig_i)_{i=1}^\bign\in\real^\bign$. 
If a strategy profile $(\sstrat,\rstrat)$ communicates perspective $\View$, then it yields expected R payoff $$
\expec{\ur}=
\sum_{i=1}^\litn \elem_i^\top\simmat{\View}{\Viewb}\reig
= \sum_{j=1}^\bign \reig_j \sum_{i=1}^\litn \left(\bbE[\state^{\View}_{i}\state^{\Viewb}_{j}
]\right)^2.
$$
If $\View$ is an $\slo$-perspective with associated weights $(\seig_i)_{i=1}^\bign$, then the strategy profile yields expected S payoff
$\expec{\us}=
\sum_{i=1}^\litn \seig_i$.
\end{lemma}
\begin{proof}
Recall that $\View$ being a perspective means $\View^\top\mvar\View=I$; and note that $\normproj:=\brac{\mathbf1_{i=j\leq\litn}}_{i,j}\in\matnn$ has $\normproj^\top\normproj=\normproj$. So letting $\somemat:=\normproj\cormat{\View}{\Viewb}$, matrix $\strats:=\rstrat\sstrat$ has
$$
\Viewb^\top\strats\mvar\Viewb
=\Viewb^\top\inv{(\View^\top)}(\normproj) \View^\top\mvar\Viewb
=\Viewb^\top(\mvar\View)(\normproj^\top\normproj) \View^\top\mvar\Viewb
=\somemat^\top\somemat.
$$
Any $j\in[\bign]$ therefore has\footnote{Recall, any $\somen\in\mathbb{Z}_+$ has $[\somen]=\{1,\ldots,\somen\}=\{i\in\mathbb{N}:\ i\leq\somen\}$.}
$$
\Tr(\Viewb\elem_j\elem_j^\top\Viewb^\top\strats\mvar)=\elem_j^\top\Viewb^\top\strats\mvar\Viewb\elem_j=(\somemat\elem_j)^\top(\somemat\elem_j)=\sum_{i=1}^\bign \somemat_{ij}^2
=\sum_{i=1}^\litn \simmat{\View}{\Viewb}_{ij}.
$$
Now, because $\Viewb$ is an $\rlo$-perspective with weights $(\reig_i)_{i=1}^\bign$, we can write $$\rlo=\sum_{j=1}^\bign \reig_j\Viewb\elem_j\elem_j^\top\Viewb^\top.$$
Recalling from \autoref{lem: trace payoff calculation} that $\expec{\ur}=\Tr\paren{\rlo\strats\mvar}$, we therefore have
$$\expec{\ur}=\sum_{j=1}^\bign\reig_j\Tr(\Viewb\elem_j\elem_j^\top\Viewb^\top\strats\mvar)
=\sum_{j=1}^\bign\reig_j\sum_{i=1}^\litn \simmat{\View}{\Viewb}_{ij}
=\sum_{i=1}^\litn\elem_i^\top\simmat{\View}{\Viewb}\reig.$$

Finally, in the case $\View$ is an $\slo$-perspective with associated weights $(\seig_i)_{i=1}^\bign$, identical calculations show that $\expec{\us}=\sum_{i=1}^\litn\elem_i^\top\simmat{\View}{\View}\seig=\sum_{i=1}^\litn\elem_i^\top I\seig=\sum_{i=1}^\litn\seig_i.$
\end{proof}

\begin{proof}[Proof of \autoref{prop: eqlbm payoffs}]
Apply \autoref{lem: similarity payoff calculation} to equilibria.
\end{proof}

\begin{definition}
An \textbf{$\rlo$-extremal $\slo$-perspective} is an $\slo$-perspective $\View$ such that, for any $\litn\in\{1,\ldots,\bign-1\}$ and $\slo$-perspective $\Viewt$, letting $(\spay,\rpay)$ be the  payoff pair from communicating perspective $\Viewt$ in the capacity-$\litn$ game, there exist $\slo$-perspectives $\bar\View$ and $\underline\View$ that are issue-equivalent to $\View$ such that the payoff pairs $(\bar\spay,\bar\rpay)$ and $(\underline\spay,\underline\rpay)$ induced by communicating $\bar\View$ and $\underline\View$, respectively, have $\underline\spay=\spay=\bar\spay$ and $\underline\rpay\leq\rpay\leq\bar\rpay$.
\end{definition}

\begin{lemma}\label{lem: set of S perspectives}
Let $\View$ be an $\slo$-perspective with associated weights $(\seig_i)_{i=1}^\bign$, and define the partition $\indsetset$ of $[\bign]$ with $i,j\in[\bign]$ in the same cell of $\indsetset$ if and only if $\seig_i=\seig_j$. Then, for
$$
\matset:=\curlyb{\somemat\in\matnn:\ \somemat \text{ orthogonal, } \somemat_{ij}=0 \text{ for $i$ and $j$ in distinct $\indsetset$-cells}},
$$
the set of $\slo$-perspectives is equal to $\curlyb{\View\somemat^\top\perm:\ \somemat\in\matset, \ \perm\in\matnn \text{ permutation matrix}}.$
Moreover, S attains the same the expected payoff from communicating perspective $\View\somemat^\top\perm$ or perspective $\View\perm$, for any $\somemat\in\matset$ and permutation matrix $\perm$. Finally,
some $\slo$-perspective and $\rlo$-perspective exist. 
\end{lemma}
\begin{proof}
Recall, $\View=[\view_1\cdots\view_\bign]$ is a perspective if and only if $\view_1,\ldots,\view_\bign$ is a $\ip\cdot\cdot_\mvar$-orthonormal basis. Moreover, left multiplication by $\slob:=\slo\mvar$ is a self-adjoint operator with respect to $\ip\cdot\cdot_\mvar$, and $\View$ is an $\slo$-perspective if and only if the $\ip\cdot\cdot_\mvar$-orthonormal basis $\view_1,\ldots,\view_\bign$ consists of eigenvectors of $\slob$. In this case, the weights $(\seig_i)_{i=1}^\bign$ associated with perspective $\View$ are exactly the eigenvalues associated with the eigenbasis: $$\ip{\somevec}{ \slob\somevec}_\mvar=\sum_{i=1}^\bign \seig_i \ip{\somevec}{\view_i}_\mvar^2 \text{ for all } \somevec\in\real^\bign.$$

The result now follows from the spectral theorem. The $\slob$-eigenspaces are pairwise $\ip\cdot\cdot_\mvar$-orthogonal and span $\real^\bign$, and $\View$ is an $\slo$-perspective and only if its columns consist of orthonormal bases of these eigenspaces. Each $\indset\in\indsetset$ corresponds to an eigenspace $\text{span}\{\view_j\}_{j\in\indset}$ whose associated eigenvalue we call (in a mild abuse) $\seig_\indset$. Then some $\ip\cdot\cdot_\mvar$-orthonormal basis $\{\view_i:\ i\in\indset\}$ exists for the $\seig_\indset$ eigenspace, and the set of all $\ip\cdot\cdot_\mvar$-orthonormal basese for it is the set of orthogonal linear combinations of $\{\view_i:\ i\in\indset\}$. Hence, an $\slo$-perspective $\View$ exists, and the set of all $\slo$-perspectives is the set of all column permutations of eigenspace-by-eigenspace orthogonal transformations of $\View$. Moreover, communicating \emph{any} perspective $\litn_\indset$ of whose first $\litn$ columns belong to the $\seig_\indset$ eigenspace for each $\indset\in\indsetset$ yields (given \autoref{lem: similarity payoff calculation}) S expected payoff $\sum_{\indset\in\indsetset}\litn_\indset\seig_\indset$. Therefore, the expected payoff from communicating $\View\somemat^\top\perm$ for $\somemat\in\matset$ and permutation matrix $\perm$ is invariant to $\somemat$. 

Finally, the existence of an $\rlo$-perspective follows from applying the above analysis to $\rlo$ instead of $\slo$.
\end{proof}

\begin{lemma}\label{lem: distinct weights means unique perspective}
Given the $\bign$ weights associated with some $\slo$-perspective, a unique $\slo$-perspective exists if and only if these $\bign$ weights are distinct.
\end{lemma}
\begin{proof}
This lemma follows directly from \autoref{lem: set of S perspectives}. If the weights are all distinct, then the set $\matset$ (as described in the statement of that lemma) consists of all diagonal matrices with diagonal entries in $\{1,-1\}$; and if the weights are not all distinct, then $\matset$ is infinite.
\end{proof}

\begin{lemma}\label{lem: extremal perspective exists}
Some $\rlo$-extremal $\slo$-perspective exists. 
\end{lemma}
\begin{proof}
Given \autoref{lem: set of S perspectives}, some $\slo$-perspective $\View$ and 
$\rlo$-perspective $\Viewb$ exist. Let $\indsetset$ and $\matset$ be as defined in the statement of that lemma for the set of $\slo$-perspectives. Observe, a matrix $\somemat\in\matnn$ belongs to $\matset$ if and only if there exist $(\somemat_\indset)_{\indset\in\indsetset}$ such that $\somemat$ is block diagonal with $\indset\times\indset$ block $\somemat_\indset$ for each $\indset\in\indset$. 

Let us set up some useful notation. Let $\reig=(\reig_i)_{i=1}^\bign\in\real^\bign$ be the weights associated with $\rlo$-perspective $\Viewb$, and let $\rlob\in\matnn$ denote the diagonal matrix with diagonal $\reig$. Let $\someorth:=\cormat{\View}{\Viewb}=\View^\top\mvar\Viewb$, which is orthogonal. For any $\indset\in\indsetset$, let $\someorth_\indset\in\mat{\indset}{\bign}$ be the submatrix of $\someorth$ consisting of its $\indset$ rows, and define the symmetric matrix $\rlob_\indset:=\someorth_\indset\rlob\someorth_\indset^\top\in\mat\indset\indset$. The spectral theorem delivers 
orthogonal $\somematt_\indset\in\mat\indset\indset$ such that $\somematt_\indset\rlob_\indset\somematt_\indset^\top$ is diagonal with diagonal vector $\somevect_\indset\in\real^\indset$. Let $\somevect=(\somevect_\indset)_{\indset\in\indsetset}\in\real^\bign$, let $\somematt\in\matset$ have its $\indset\times\indset$ block equal to $\somematt_\indset$ for each $\indset\in\indsetset$, and define the $\slo$-perspective $\Viewt:=\View\somematt^\top$. We will show that the $\Viewt$ is $\rlo$-extremal.

Consider any $\indset\in\indsetset$ and $\somemat\in\matset$ with $\indset\times\indset$ block $\somemat_\indset$. By construction, the $\indset$ rows of $\somemat\someorth$ are given by $\somemat_\indset\someorth_\indset$. Therefore, any $i\in\indset$ has \begin{eqnarray*}
\brac{\simmat{\View\somemat^\top}{\Viewb}\reig}_i
&=&\brac{\cormat{\View\somemat^\top}{\Viewb}\ \rlob\ \cormat{\View\somemat^\top}{\Viewb}^\top}_{ii} \\
&=&\brac{\somemat\someorth\rlob(\somemat\someorth)^\top
}_{ii} \\
&=&\brac{\somemat_\indset\someorth_\indset\rlob(\somemat_\indset\someorth_\indset)^\top
}_{ii} \\
&=&\brac{(\somemat_\indset\somematt_\indset^\top)(\somematt_\indset\rlob_\indset\somematt_\indset^\top)(\somemat_\indset\somematt_\indset^\top)^\top
}_{ii} \\
&=& (\somebi_\indset\somevect)_i,
\end{eqnarray*}
where $\somebi_\indset\in\mat\indset\indset$ is the entrywise square of the orthogonal matrix $\somemat_\indset\somematt_\indset^\top$. 
Hence, $[\simmat{\View\somemat^\top}{\Viewb}\reig]_\indset=\somebi_\indset\somevect_\indset$ is majorized by $\somevect_\indset$ by the Schur-Horn theorem \citep[Theorem 2.B.6,][]{marshall2011inequalities}, and specializing this calculation to $\somemat_\indset=\somematt_\indset$ tells us $[\simmat{\Viewt}{\Viewb}\reig]_\indset=\somevect_\indset$. So $[\simmat{\Viewt}{\Viewb}\reig]_\indset$ majorizes $[\simmat{\View\somemat^\top}{\Viewb}\reig]_\indset$ for every $\somemat\in\matset$. 

Now, take any $\somemath\in\matset$ and permutation matrix $\permh$, and consider the strategy profile (for capacity $\litn$) that communicates perspective $\View\somemath^\top\permh$---\autoref{lem: set of S perspectives} tells us any $\slo$-perspective is of this form. We want to find some $\slo$-perspective, issue equivalent to $\View\somematt^\top\permt$, such that communicating it yields the same S payoff as $\View\somemath^\top\permh$ and a weakly higher [resp. lower] R payoff than $\View\somemath^\top\permh$. 
To that end, let $\Viewh:=\View\somemath^\top$ and $\somevech:=\simmat{\Viewh}{\Viewb}\reig$, and let $\indsetb$ be the set of all $i\in[\bign]$ such that some $j\in[\litn]$ has $\permh_{ij}=1$. 
Let $\perm\in\matset$ be a permutation matrix such that, for each $\indset\in\indsetset$, every entry of $(\perm\somevect)_{\indset\cap\indsetb}$ is weakly higher [resp. lower] than every entry of $(\perm\somevect)_{\indset\setminus\indsetb}$; and let $\permh:=\perm\permt$. Then,
$$
\sum_{i=1}^\litn \elem_i^\top\paren{\permt^\top\somevect-\permh^\top\somevech}
= \sum_{i=1}^\litn (\permh\elem_i)^\top \paren{\perm\somevect-\somevech}\\
=\sum_{\indset\in\indsetset}\ \sum_{i\in\indset\cap\indsetb}\paren{\perm\somevect-\somevech}_i,$$
which is nonnegative [resp. nonpositive] because $\somevect_\indset$ majorizes $\somevech_\indset$ for every $\indset\in\indsetset$.

Finally,  let us see that $\Viewt\permt$ 
is an $\slo$-perspective that satisfies the desired payoff comparisons.
First, observe $\Viewt\permt=(\View\somematt^\top)(\perm^\top\permh)=\View(\perm\somematt)^\top\permh$ and $\perm\somematt\in\matset$. Hence,  \autoref{lem: set of S perspectives} tells us $\Viewt\permt$ is an $\slo$-perspective, and that communicating it yields the same expected S payoff as communicating $\Viewh\permh=\View\somemath^\top\permh$. 
Toward the R payoff comparison,  \autoref{lem: similarity payoff calculation} says that sharing perspective $\Viewh\permh$ yields an expected R payoff of
$$
\sum_{i=1}^\litn \elem_i^\top\simmat{\Viewh\permh}{\Viewb}\reig
=\sum_{i=1}^\litn \elem_i^\top\permh^\top\simmat{\Viewh}{\Viewb}\reig
=\sum_{i=1}^\litn \elem_i^\top\permh^\top \somevech,
$$
and an identical calculation says sharing perspective $\Viewt\permt$ yields an expected R payoff of $\sum_{i=1}^\litn \elem_i^\top\permt^\top \somevect$. The previous paragraph establishes the desired ranking of these two quantities, and so $\Viewt\permt$ is as required. Therefore, $\Viewt$ is $\rlo$-extremal.
\end{proof}

\begin{lemma}\label{lem: schur-horn translated to perspectives}
Fix a vector $\someeigen\in\real^\bign$ and a perspective $\Viewb$. Given $\somevec\in\real^\bign$, vector $\someeigen$ majorizes $\somevec$ if and only if a perspective $\View$ exists such that $\simmat{\View}{\Viewb}\someeigen=\somevec$.\end{lemma}
\begin{proof}
Let us observe that the set of possible $\simmat{\View}{\Viewb}$ matrices, as we range over all perspectives $\View$, coincides with the set of all orthostochastic matrices---that is, the entrywise squares of orthogonal matrices. Indeed, defining the orthogonal matrix $\Viewbt:=\rmvar\Viewb$, the former set is
\begin{eqnarray*}
\curlyb{\cormat{\View}{\Viewb}:\ \View\text{ a perspective}}
&=& \curlyb{\View^\top\mvar\Viewb:\ \View\in\matnn \text{ has } \View^\top\mvar\View=I} \\
&=& \curlyb{\View^\top\rmvar\Viewbt:\ \View\in\matnn \text{ has } \rmvar\View \text{ orthogonal}} \\
&=& \curlyb{\Viewt^\top\Viewbt:\ \Viewt\in\matnn \text{ orthogonal}},
\end{eqnarray*}
which is equal to the set of orthogonal matrices because the latter forms a group under multiplication. So $\curlyb{\simmat{\View}{\Viewb}:\ \View\text{ a perspective}}$ is the set of orthostochastic matrices.

Therefore, $\somevec$ admits a perspective $\View$ with $\simmat{\View}{\Viewb}\someeigen=\somevec$ if and only if $\somevec=\somemat\someeigen$ for some orthostochastic matrix $\somemat$. But then the Schur-Horn theorem \citep[Theorem 2.B.6,][]{marshall2011inequalities} shows the latter condition is equivalent to $\someeigen$ majorizing $\somevec$.
\end{proof}

\begin{lemma}\label{lem: possible R payoff sets}
Fix an R loss matrix $\rlo$, and let $\reig=(\reig_i)_{i=1}^\bign$ be the weights associated with some $\rlo$-perspective. Given $\epset\subseteq\real$, the following are equivalent:
\begin{enumerate}[(i)]
    \item\label{lem: possible R payoff sets, item perspective} Some perspective $\View$ exists such that $
    \epset$ is the set of R payoffs associated with strategy profiles communicating perspectives that are issue-equivalent to $\View$. (Hence, $\epset$ is finite and nonempty.)
    \item\label{lem: possible R payoff sets, item extremal} Condition \eqref{lem: possible R payoff sets, item perspective} holds, and is witnessed by a perspective that is the unique $\slo$-perspective (which is therefore $\rlo$-extremal) for some S loss $\slo$.

    \item\label{lem: possible R payoff sets, item majorization} 
    Some $\issur\in\real^\bign$ majorized by $\reig$ has $\epset=\curlyb{\sum_{i\in\indset}\issur_i:\ \indset\subseteq[\bign] \text{ has } |\indset|=\litn }.$
\end{enumerate}
Moreover, given any S loss matrix $\slo$: the highest, lowest Pareto-optimal, and lowest equilibrium R payoffs all come from some set $\epset$ satisfying the above equivalent conditions.
\end{lemma}
\begin{proof}
First, let us see \eqref{lem: possible R payoff sets, item perspective} is equivalent to its strengthening \eqref{lem: possible R payoff sets, item extremal}. We need only verify that any perspective $\View$ can be the unique $\slo$-perspective for some S loss $\slo$. And indeed, the S loss matrix given by $\slo=\sum_{j=1}^\bign j\View\elem_j\elem_j^\top\View^\top$ admits $\View$ as an $\slo$-perspective with associated weights $(j)_{j=1}^\bign$. Thus, by \autoref{lem: distinct weights means unique perspective}, this $\slo$-perspective is unique, hence $\rlo$-extremal.

Next, let us observe \eqref{lem: possible R payoff sets, item perspective} is equivalent to \eqref{lem: possible R payoff sets, item majorization}. To see it, let $\Viewb$ be the $\rlo$-perspective with associated weights $(\reig_i)_{i=1}^\bign$. \autoref{lem: similarity payoff calculation} tells us \eqref{lem: possible R payoff sets, item perspective} is equivalent to the existence of some perspective $\View$ such that $\issur=\simmat{\View}{\Viewb}\reig$ has
$\epset=\curlyb{\sum_{i\in\indset}\issur_i:\ \indset\subseteq[\bign] \text{ has } |\indset|=\litn }.$ This condition is then equivalent to \eqref{lem: possible R payoff sets, item majorization} by \autoref{lem: schur-horn translated to perspectives}.

Having established the three-way equivalence, we need only prove the last assertion. To that end, fix some S loss $\slo$. Let $\View$ be some $\rlo$-extremal $\slo$-perspective (which exists by \autoref{lem: extremal perspective exists}). \autoref{prop: eqm char} tells us all profiles that communicate a perspective issue-equivalent to $\View$ are equilibria. \autoref{prop: eqm char} also says every equilibrium communicates some $\slo$-perspective, and so weakly Pareto dominates one of these equilibria and so is weakly Pareto dominated by one of these equilibria. The claim follows.
\end{proof}

\begin{proof}[Proof of \autoref{prop: receiver payoff set}]

Recall, for loss matrix $\rlo$, the vector $(\reig_i)_{i=1}^\bign\in\real^\bign$ records the weights on the $\bign$ issues according an $\rlo$-perspective $\Viewb$. Permuting columns of $\Viewb$, we may assume $\reig_1\geq\cdots\geq\reig_\bign$. Given \autoref{lem: possible R payoff sets}, condition \eqref{prop: receiver payoff set, payoff part} holds if and only if some $\issur\in\real^\bign$ majorized by $\reig$ has \begin{equation}\label{eqn: relate high and low payoffs to picking entries of fixed vector}
    \max_{\indset\subseteq[\bign]:\ |\indset|=\litn}\ \sum_{i\in\indset}\issur_i=\highur \ \text{ and }\ \min_{\indset\subseteq[\bign]:\ |\indset|=\litn}\ \sum_{i\in\indset}\issur_i=\lowur.
\end{equation}
This condition is invariant to permuting the entries of $\issur$, so we can without loss of generality focus on $\issur$ with $\issur_1\geq\cdots\geq\issur_\bign$. In this case, the majorization constraint on $\issur$ says $\sum_{i=1}^\somen(\reig_i-\issur_i)$ is nonnegative for every $\somen\in[\bign]$ and zero for $\somen=\bign$; and condition~\eqref{eqn: relate high and low payoffs to picking entries of fixed vector} reduces to $\highur=\sum_{i=1}^\litn\issur_i$ and $\lowur=\sum_{i=\bign-\litn+1}^\bign\issur_i$. Now, let $\litnt:=\min\{\litn,\bign-\litn\}$. Replacing the first $\litnt$ entries of $\issur$ with their average, replacing the last $\litnt$ entries of $\issur$ with their average, and (if $\litn\neq\tfrac\bign2$) replacing all other entries with their average preserves the majorization constraint (and the property that $\issur_1\geq\cdots\geq\issur_\bign$) and does not affect condition~\eqref{eqn: relate high and low payoffs to picking entries of fixed vector}. We can therefore restrict attention to $\issur$ that are weakly decreasing, and constant on each of these three sets of indices. Moreover, note that a $\issur$ of this form is majorized by $\reig$ if and only if it satisfies 
\begin{equation}\label{eqn: simplified majorization}
\sum_{i=1}^\bign\issur_i=\sum_{i=1}^\bign\reig_i \text{ and } \sum_{i=1}^\somen\issur_i\leq\sum_{i=1}^\somen\reig_i \text{ for each } \somen\in\{\litn,\bign-\litn\}.
\end{equation}
We can further simplify \eqref{eqn: simplified majorization} by substituting in \eqref{eqn: relate high and low payoffs to picking entries of fixed vector}. To that end, let $\medur$ denote $\litn$ times the average value of all of the $\issur$ entries other than the first $
\litnt$ and last $\litnt$ if $\litn\neq\tfrac\bign2$, and let $\medur$ be arbitrary otherwise. Then the equality constraint \eqref{eqn: relate high and low payoffs to picking entries of fixed vector} is equivalent to $\lowur+\highur+\tfrac{\bign-2\litn}{\litn}\medur=\tfrac{\bign}{\litn}\medurb$, and (given this equality constraint) the two inequality constraints are equivalent to $\lowur\geq\lowurb$ and $\highur\leq\highurb$.\footnote{To see this form of the equality constraint, observe that $\tfrac{\bign}{\litn}\medurb$ is the sum of all entries of $\reig$ and $\lowur+\highur$ is the sum of the first $\litn$ and last $\litn$ entries of $\issur$. The term $\tfrac{\bign-2\litn}{\litn}\medur$ can be any sign, and adjusts the sum of $\issur$ entries for the $|\bign-2\litn|$ entries in the middle that are uncounted or double-counted.}

Summarizing the conclusion of the previous paragraph, we have shown a pair $(\highur,\lowur)\in\real^2$ satisfies \eqref{prop: receiver payoff set, payoff part} 
if and only if some $\medur$ has $\lowurb\leq\lowur\leq\medur\leq\highur\leq\highurb$ and $\lowur+\highur+\tfrac{\bign-2\litn}{\litn}\medur=\tfrac{\bign}{\litn}\medurb$. But now, a given $\lowur\leq\highur$ admit some $\medur$ between them such that $\lowur+\highur+\tfrac{\bign-2\litn}{\litn}\medur=\tfrac{\bign}{\litn}\medurb$ if and only if $\tfrac{\bign}{\litn}\medurb$ is between $\lowur+\highur+\tfrac{\bign-2\litn}{\litn}\lowur$ and $\lowur+\highur+\tfrac{\bign-2\litn}{\litn}\highur$. 
Rearranging, this condition is the same as $\medurb$ lying between $\tfrac\litn\bign\lowur+\paren{1-\tfrac\litn\bign}\highur$ and $\tfrac\litn\bign\highur+\paren{1-\tfrac\litn\bign}\lowur$---which also requires that $\medurb$ lie between $\lowur$ and $\highur$. The proposition follows.
\end{proof}

\begin{proof}[Proof of \autoref{cor: maximizing worst minimizes best}]
The equivalence of the first two conditions holds because the range of $\scalar$ specified in \autoref{prop: receiver payoff set} belongs to $(0,1)$. 

Now, for any perspective $\View$, the vector $\simmat{\View}{\Viewb}\reig$ has its entries averaging to $\tfrac1\litn\medurb$ because $\simmat{\View}{\Viewb}$ is stochastic. The equivalence of the third item and (say) the first therefore follows directly from \autoref{prop: eqlbm payoffs}.    
\end{proof}

\begin{definition}\label{def: genericity}
Given a subset $\vectset$ of a Euclidean space, say a subset of $\vectset$ is \textbf{generic} in $\vectset$ if it contains an open and dense subset $\vectset$, whose complement in the affine span of $\vectset$ is Lebesgue null in that affine span. 
\end{definition}

\begin{notation}
Let $\Loss$ denote the set of positive definite matrices, a convex open set of the Euclidean space of pairs of symmetric $\bign\times\bign$ matrices. Say a property of matrix $\rlo$ or $\slo$ [resp. pairs $(\slo,\rlo)$] holds \textbf{generically} if the set of matrices [resp. pairs] for which it holds is generic in $\Loss$ [resp. $\Loss^2$].
\end{notation}

\begin{remark}
Note that the definition of genericity in \autoref{def: genericity} is stronger than measure theoretic genericity (i.e., its complement being Lebesgue null in the ambient affine space) and stronger than topological genericity (containing a countable intersection of open and dense sets). It follows that all our results that invoke generic parameter choices will continue to hold if one replaces our definition with either of these notions of genericity.
\end{remark}

\begin{lemma}\label{lem: semialgebraic and dense implies generic}
If $\vectset$ is a semialgebraic subset of some Euclidean space, then any semialgebraic dense subset of $\vectset$ is generic in $\vectset$. In particular, any semialgebraic dense subset of $\Loss$ or $\Loss^2$ is generic.
\end{lemma}
\begin{proof}
Let $\vectset_1$ be any dense semialgebraic set in $\vectset$, and let $\vectset_0:=\vectset\setminus\vectset_1$. Throughout this proof, we refer to definitions and results from \cite{van1998tame}---hereafter vdD98. 
By cell decomposition (Proposition 3.2.11 of vdD98, which applies by Corollary 2.2.11 of vdD98), each of $\vectset_1$ and $\vectset_0$ can be represented as a disjoint union of finitely many  cells (as in Definition 2.2.3 of vdD98) defined by polynomial functions. As noted in vdD98 immediately after this definition, a semialgebraic set has the same dimension (in the sense of Definition 4.1.1 from vdD98) as the underlying ambient linear space if and only if it contains a full-dimensional cell. 
Hence, because $\vectset_1$ is dense in $\vectset$ and semialgebraic, it must be that all of the cells comprising $\vectset_0$ have strictly lower dimension than $\vectset$. 
Each low-dimensional cell in the decomposition therefore has low-dimensional closure too by vdD98's Theorem 4.1.8, and so is nowhere dense (because it has empty interior and is open in its closure) and Lebesgue null \citep[Lemma 7.25][]{rudin1987real}. 
The union of all full-dimensional cells in the decomposition is then a subset of $\vectset_1$ as required for the first assertion.

To see the second assertion, it then remains to note the convex set $\Loss$ (hence also $\Loss^2$) is semialgebraic. And indeed, $\slo$ is positive definite matrix if and only if $\somevec^\top\slo\somevec>0$ for every nonzero $\somevec\in\real^\bign$.
\end{proof}

\begin{lemma}\label{lem: generically distinct weights}
For a generic loss matrix $\slo$, the $\bign$ weights associated with some $\slo$-perspective are distinct.
\end{lemma}
\begin{proof}
Let $\Loss_1\subseteq\Loss$ denote the set of $\slo\in\Loss$ with distinct weights. Because $\Loss_1$ is semialgebraic, \autoref{lem: semialgebraic and dense implies generic} tells us it suffices to see an arbitrary $\slo\in\Loss$ is a limit of members of $\Loss_1$. To that end, take some $\slo$-perspective $\View$ (which exists by \autoref{lem: set of S perspectives}), witnessed by weights $\seig_1,\ldots,\seig_\bign$. We can then write $\slo=\sum_{j=1}^\bign \seig_j\View\elem_j\elem_j^\top\View^\top.$
For $\varepsilon>0$, define the positive definite matrix $$\slo^\varepsilon=\sum_{j=1}^\bign (\seig_j+j\varepsilon)\View\elem_j\elem_j^\top\View^\top,$$
which converges to $\slo$ as $\varepsilon\to0$. By construction, $\View$ is a $\slo$-perspective with weights $(\seig_j+j\varepsilon)_{j=1}^\bign$. Meanwhile, whenever $\varepsilon>0$ is small enough, the $\bign$ numbers $(\seig_j+j\varepsilon)_{j=1}^\bign$ are all distinct, so that $\slo^\varepsilon\in\Loss_1$.
\end{proof}

\begin{proof}[Proof of \autoref{cor: payoff set dominance}]
Let $\Viewb$ is an $\rlo$-perspective with weights $(\reig_i)_{i=1}^\bign$, and let $\lowurb,\medurb,\highurb$ be as defined in the preamble to \autoref{prop: receiver payoff set}. Say a pair $(\highur,\lowur)$ is payoff feasible if it satisfies the two equivalent conditions of that proposition---or equivalently, some S loss $\slo$ exists such that $\highur$ and $\lowur$ are R's best and worst equilibrium payoffs, respectively. In what follows, we detail the precise conditions for each of \eqref{cor: payoff set dominance, item same best} and \eqref{cor: payoff set dominance, item strict Pareto} to hold, and then we argue that both conditions hold for generic $\rlo$.\footnote{The proof of the genericity result can be simplified by not characterizing the exact conditions for each of \eqref{cor: payoff set dominance, item same best} and \eqref{cor: payoff set dominance, item strict Pareto}. We present the current version for completeness.}

First, if $\litn=\tfrac12\bign$, then \autoref{prop: receiver payoff set}  tells us any payoff-feasible $(\highur,\lowur)$ has $\tfrac12\highur+\tfrac12\lowur=\medurb$, and so no two payoff-feasible pairs are coordinate-wise ranked, meaning neither \eqref{cor: payoff set dominance, item same best} nor \eqref{cor: payoff set dominance, item strict Pareto} holds.

Now, using \autoref{prop: receiver payoff set}, let us see \eqref{cor: payoff set dominance, item same best} is equivalent to having $\litn\neq\tfrac\bign2$ and $\reig$ not being a constant vector. 
Given the previous paragraph, we can focus on the case of $\litn\neq\tfrac12\bign$. If $\reig$ is a constant vector, then $(\medurb,\medurb)=(\highurb,\lowurb)$ is the unique payoff-feasible pair, so \eqref{cor: payoff set dominance, item same best} fails. 
Conversely, if $\reig$ is not a constant vector, then  $\lowurb<\medurb<\highurb$, and $\co\curlyb{\tfrac{\litn}{\bign},1-\tfrac{\litn}{\bign}}$ contains a neighborhood of $\tfrac12$. So take any $\varepsilon>0$ small enough that $\varepsilon<\medurb-\lowurb$ and $\varepsilon<\highurb-
\medurb$. Then, any payoff pair in a small enough neighborhood of $(\medurb+\varepsilon,\medurb-\varepsilon)$ is payoff feasible, delivering \eqref{cor: payoff set dominance, item same best}.

Next, we argue that \eqref{cor: payoff set dominance, item strict Pareto} is equivalent to either having $\litn<\tfrac\bign2$ and the last $\bign-\litn$ entries of $\reig$ all coinciding, or having $\litn>\tfrac\bign2$ and the first $\litn$ entries of $\reig$ all coinciding. Given the paragraph before the previous one, we can focus on the case of $\litn\neq\tfrac12\bign$. Now, let $\litnt:=\min\{\litn,\bign-\litn\}$. 
Then, by \autoref{prop: receiver payoff set}, a given $\lowur\in\real$ is such that $(\highurb,\lowur)$ is payoff feasible if and only if 
$\lowur\in[\lowurb,\medurb]$ and some $\scalar\in\brac{\tfrac{\litnt}{\bign},1-\tfrac{\litnt}{\bign}}$ has $\medurb=(1-\scalar)\lowur+\scalar\highurb$. 
So \eqref{cor: payoff set dominance, item same best} holds if and only if some $\lowur>\lowurb$ satisfies these conditions. We know $(\highurb,\lowurb)$ is payoff feasible because (\autoref{prop: eqlbm payoffs}) the best and worst equilibrium payoffs for sender $\slo=\rlo$ are $\highurb$ and $\lowurb$. We can then raise $\lowurb$ and maintain payoff feasibility if and only if $\scalar$ can be correspondingly lowered to maintain $\medurb=(1-\scalar)\lowur+\scalar\highurb$. So \eqref{cor: payoff set dominance, item same best} holds if and only if \begin{equation*}
    \medurb>\paren{1-\tfrac{\litnt}{\bign}}\lowurb+\tfrac{\litnt}{\bign}\highurb.
\end{equation*}
If $\litn<\tfrac{\bign}2$ [resp. $\litn>\tfrac{\bign}2$] then this inequality rearranges to say that the average of the last $\bign-\litn$ [resp. first $\litn$] entries of $\reig$ is strictly greater than the average of the last $\litn$ [resp. first $\bign-\litn$] entries of $\reig$; this condition is  in turn equivalent to the last $\bign-\litn$ [resp. first $\litn$] entries of $\reig$ all coinciding.

Finally, we turn to the genericity claim. If $\litn\neq\tfrac{\bign}2$, then $\max\curlyb{\litn,\bign-\litn}\geq2$. Hence, given the above characterizations, failure of either of \eqref{cor: payoff set dominance, item same best} or \eqref{cor: payoff set dominance, item strict Pareto} requires $\reig$ to have at least two entries in common. The corollary then follows from applying \autoref{lem: generically distinct weights} to $\rlo$.
\end{proof}

\begin{proof}[Proof of \autoref{prop: ranking sender perspectives}]
In what follows, let $\somemat_1:=\simmat{\View_1}{\Viewb}$ and $\somemat_2:=\simmat{\View_2}{\Viewb}$; and let $\wtset:=\curlyb{\wt\in\{0,1\}^\bign:\ \sum_{i=1}^\bign\wt_i=\litn}$, the set of weight vectors that put weight $1$ on $\litn$ components and $0$ on $\bign-\litn$. 

By assumption, for either $\ell\in\{1,2\}$, the matrix $\slo_\ell$ admits a unique perspective, meaning any perspective is obtained from $\View_\ell$ by permuting columns and multiplying some by $-1$. Multiplying a column of perspective $\View$ by $-1$ leaves $\simmat{\View}{\Viewb}$ unchanged, whereas permuting columns of $\View$ permutes the rows of $\simmat{\View}{\Viewb}$ in the same way. Therefore, in light of \autoref{prop: eqm char} and \autoref{lem: similarity payoff calculation}, if the weights associated with $\rlo$-perspective $\Viewb$ are $(\reig_i)_{i=1}^\bign\in\real^\bign$, then $\goodur(\slo_\ell,\rlo)=\max_{\wt\in\wtset} \curlyb{\wt^\top\somemat_\ell\reig}$ and $
\badur(\slo_\ell,\rlo)=\min_{\wt\in\wtset} \curlyb{\wt^\top\somemat_\ell\reig}=-\max_{\wt\in\wtset} \curlyb{-\wt^\top\somemat_\ell\reig}.
$

Now, say $\slo_1\succsim_+\slo_2$ if every $\rlo$ with perspective $\Viewb$ has $\goodur(\slo_1,\rlo)\geq\goodur(\slo_2,\rlo)$; and say $\slo_1\succsim_-\slo_2$ if every $\rlo$ with perspective $\Viewb$ has $\badur(\slo_1,\rlo)\leq\badur(\slo_2,\rlo)$. Observe,
\begin{eqnarray*}
\slo_1\succsim_{\pm}\slo_2
&\iff& \text{every } \reig\in\real^\bign_{++} \text{ has } \max_{\wt\in\wtset} \curlyb{\pm\wt^\top\somemat_1\reig} \geq \max_{\wtb\in\wtset} \curlyb{\pm \wtb^\top\somemat_2\reig}\\
&\iff& \text{every } \wtb\in\wtset \text{ has } \inf_{\reig\in\real^\bign_{++}} \max_{\wt\in\wtset} \curlyb{\pm\paren{\wt^\top\somemat_1- \wtb^\top\somemat_2}\reig}\geq0\\
&\iff& \text{every } \wtb\in\wtset \text{ has } \max_{\wt\in\co\wtset} \inf_{\reig\in\real^\bign_{++}}\curlyb{\pm\paren{\wt^\top\somemat_1- \wtb^\top\somemat_2}\reig}\geq0\\
&\iff& \text{every } \wtb\in\wtset \text{ admits } \wt\in\co\wtset \text{ with } \pm\paren{\wt^\top\somemat_1- \wtb^\top\somemat_2}\geq0,
\end{eqnarray*}
where the third equivalence follows from Sion's minimax theorem.
But next note that the matrices $\somemat_1,\somemat_2$ are bistochastic, hence row stochastic. Therefore, letting $\elem\in\real^\bign$ be the vector with all entries $1$, any $\wt,\wtb\in\co\wtset$ have $$\pm\paren{\wt^\top\somemat_1- \wtb^\top\somemat_2}\elem=\pm\paren{\wt^\top\elem- \wtb^\top\elem}=\pm(\litn-\litn)=0.$$
Thus, the only way the vector inequality $\pm\paren{\wt^\top\somemat_1- \wtb^\top\somemat_2}\geq0$ can hold is if the left-hand side is the zero vector, that is, $\wt^\top\somemat_1= \wtb^\top\somemat_2$.

So we have shown conditions $\slo_1\succsim_{+}\slo_2$ and $\slo_1\succsim_{-}\slo_2$ are each equivalent to the condition that every $\wtb\in\wtset$ admits $\wt\in\co\wtset$ with $\wt^\top\somemat_1= \wtb^\top\somemat_2$, as required.
\end{proof}

\begin{lemma}\label{lem: investment full solution}
Suppose $\Viewb$ is an $\rlo$-perspective with associated weights $\reig=(\reig_i)_{i=1}^\bign$, perspective $\View$ is an $\rlo$-extremal $\slo$-perspective such that $\issur:=\simmat{\View}{\Viewb}\reig$ have $\issur_1\geq\cdots\geq\issur_\bign$, and $\aissur:=\tfrac1\bign\sum_{i=1}^\bign\issur_i=\tfrac1\bign\sum_{i=1}^\bign\reig_i$. \begin{enumerate}[(i)]
    \item Under R-best equilibrium selection: capacity $\bign$ is optimal if and only if $\cost\leq\issur_\bign$, capacity $0$ is optimal if and only if $\cost\geq\issur_1$, and capacity $\litn\in\{1,\ldots,\bign-1\}$ is optimal if and only if $\issur_\litn\geq\cost\geq\issur_{\litn+1}$.
    \item Under R-worst equilibrium selection: capacity $\bign$ is optimal if and only if $\cost\leq\aissur$, capacity $0$ is optimal if and only if $\cost\geq\aissur$, and capacity $\litn\in\{1,\ldots,\bign-1\}$ is optimal if and only if $\cost=\issur_1=\cdots=\issur_\bign=\aissur$.
\end{enumerate}
\end{lemma}
\begin{proof}
For any given $\litn\in\{0,\ldots,\bign\}$, \autoref{prop: eqm char} and $\View$ being $\rlo$-extremal tell us a best and worst equilibrium for R take the form of communicating whichever perspective R most prefers among those that are issue-equivalent to $\View$. By \autoref{lem: similarity payoff calculation}, then, the best and worst R equilibrium payoff given capacity $\litn$ are $\sum_{i=1}^\litn \issur_i$ and $\sum_{i=\bign-\litn+1}^\bign \issur_i$, respectively.\footnote{Although that lemma and its inputs assume a capacity in $\{1,\ldots,\bign-1\}$, the proofs apply identically to the cases of $\{0,\bign\}$. Alternatively, the payoff formulas are trivial to prove for these special cases, given the equilibrium characterizations described in \autoref{sec: model}.}  Because $\issur_1\geq\cdots\geq\issur_\bign$, these sums are weakly concave in $\litn$ and weakly convex in $\litn$, respectively, and the latter is not affine in $\litn$ unless $\issur$ is a constant vector. Because the investment cost is linear in $\litn$, the observations of the previous sentence apply identically to the best and worst R equilibrium payoffs net of investment costs.

Consider the case of R-best equilibrium selection. Because R's net payoff is weakly concave in investment, an investment level is optimal if and only if R cannot profitably increase or decrease investment by $1$. Meanwhile, the net payoff difference between investment level $\litnt\in[\bign]$ and investment level $\litnt-1$ is $\issur_{\litnt}-\cost$, so the optimality characterization in the lemma follows.

Turn now to the case of R-worst equilibrium selection. Because R's net payoff is weakly convex in investment, one of the extreme investment levels $\{0,\bign\}$ is optimal. The payoff difference between investment levels $\bign$ and $0$ is $\bign(\aissur-\cost)$. Hence, the stated characterization of when each of $\{0,\bign\}$ is optimal follows. 
Finally, given that the objective is weakly convex on its scalar domain, no interior investment level is optimal unless this objective is constant (in which case all are optimal). But the latter condition holds if and only if $\cost=\issur_1=\cdots=\issur_\bign$, in which case this quantity also coincides with $\aissur$. The lemma follows.
\end{proof}

\begin{lemma}\label{lem: generically multiple R payoffs}
Given $\rlo$, take any $\rlo$-perspective $\Viewb$ with associated weights $(\reig_i)_{i=1}^\bign$.
Suppose $\rlo$ is not proportional to $\inv{\mvar}$ (equivalently $\{\reig_i\}_{i=1}^\bign$ do not all coincide). Then, generic $\slo\in\Loss$ has a unique $\slo$-perspective $\View$, which has all entries of $\issur=\simmat{\View}{\Viewb}\reig$ distinct.
\end{lemma}
\begin{proof}
Let $\Views\subseteq\matnn$ denote the set of all perspectives. For any distinct $i,j\in[\bign]$, define $\func_{ij}:=(\elem_i-\elem_j)^\top\simmat{\cdot}{\Viewb}\reig:\Views\to\real$.

Below, we will establish that $\func_{ij}$ is not zero on any nonempty neighborhood for any distinct $i,j\in[\bign]$. Let us first see how this would deliver the lemma. In this case, $\Views\setminus\func_{ij}^{-1}(0)$ is dense in $\Views$, and so every $\slo\in\Loss$ with unique perspective is in the closure of the set of such loss matrices with perspective in $\Views\setminus\func_{ij}^{-1}(0)$.\footnote{Holding the $\bign$ distinct weights fixed and perturbing the perspective amounts to a small perturbation of the loss matrix.} But then, \autoref{lem: distinct weights means unique perspective} and 
\autoref{lem: generically distinct weights} imply the matrices in $\Loss$ whose unique perspective lives in $\Views\setminus\func_{ij}^{-1}(0)$ is dense in $\Loss$, and so (since it is semialgebraic) generic by \autoref{lem: semialgebraic and dense implies generic}. But a finite intersection of generic sets is generic; hence, for a generic set of matrices in $\Loss$, the unique perspective is in $\Views\setminus\bigcup_{i,j\in[\bign]:\ i\neq j}\func_{ij}^{-1}(0)$. But any $\View$ in this set is such that $\simmat{\View}{\Viewb}\reig$ has all entries distinct, by definition.

So fixing distinct $i,j\in[\bign]$, it remains to see $\func=\func_{ij}$ is not zero on any neighborhood. To that end, let $\Views_0$ be the interior of $\inv{\func}(0)$ in $\Views$, which is open by construction; and let $\Views_+$ comprise the positive-determinant members of $\Views$.

Now, we argue $\Views_0$ is closed in $\Views$. Consider any $\View$ in its closure. Because $\Views$ is the zero set of the polynomial $\Viewt\mapsto\Viewt^\top\mvar\Viewt-I$ on $\matnn$, the real-analytic implicit function theorem \citep[Theorem 2.3.5,][]{krantz2002primer} yields an open neighborhood $\somenhd$ of a Euclidean space and an analytic map from said neighborhood onto a neighborhood of $\View$ in $\Views$---without loss, say $\somenhd$ is convex. The composition of this analytic map with the polynomial $\func$ is therefore analytic too \citep[Proposition 2.2.8,][]{krantz2002primer}. But because $\View$ is in the closure of $\Views_0$, every element of the convex set $\somenhd$ is on some line that intersects the open neighborhood in $\somenhd$ on which this composed function is zero. 
Applying the identity theorem \citep[Corollary 1.2.7][]{krantz2002primer}, it follows that the composed function is globally zero on $\somenhd$. Therefore, $\func$ is zero in a neighborhood of $\View$, meaning $\View\in\Views_0$ too.

Next, let us see that $\Views_+$ is either disjoint from $\Views_0$ or contained in it. Indeed, because the positive definite matrix $\sqrt{\mvar}$ has strictly positive determinant, and 
the self-map on $\matnn$ given by left-multiplication by $\sqrt{\mvar}$ is an invertible linear map that maps $\Views$ onto the orthogonal group $O(\bign,\real)$, it follows that this linear map  is a homeomorphism between $\Views_+$ and the special orthogonal group. Because the latter is connected \citep[Example 9.14,][]{baker2003matrix}, it follows that $\Views_+$ is. But in light of the previous paragraph (and the definition), $\Views_0$ is clopen, implying $\Views_0\cap\Views_+$ is clopen in $\Views_+$. As the only clopen sets in $\Views_+$ are $\Views_+$ itself and the empty set, the claim follows.

But now, observe that the self-map $\Views$ that multiplies the first column by $-1$ is a homeomorphism that maps the open set $\Views_+$ onto its open complement, and that preserves the value of $\func$. Therefore, $\Views_0$ is invariant under this self-map. Hence, the conclusion of the previous paragraph implies $\Views_0$ is either the empty set or all of $\Views$. 

Given the above paragraph, we will know $\Views_0$ is empty if we determine it is not equal to $\Views$. To see this fact, note that for any permutation matrix $\perm\in\matnn$, the perspective $\Viewb\perm$ has 
$\func(\Viewb\perm)=(\elem_i-\elem_j)^\top\perm^\top\simmat{\Viewb}{\Viewb}\reig=\brac{\perm(\elem_i-\elem_j)}^\top\reig$. But then, because $\reig$ is not a constant vector, some appropriate choice of this permutation yields $\func(\Viewb\perm)\neq0$. Hence, $\Viewb\perm\in\Views\setminus\Views_0$, delivering the lemma.
\end{proof}

\begin{proof}[Proof of \autoref{prop: investment}]
Letting $\issur,\aissur$ be as given in the statement of \autoref{lem: investment full solution}, that lemma tells us the cutoffs $(\hcost,\mcost,\lcost):=(\issur_1,\aissur,\issur_\bign)$ satisfy all three items in the list. Moreover, $\hcost\geq\mcost\geq\lcost>0$ in general, with all inequalities strict if $\issur_1>\issur_\bign$. 

It remains to see that, generically, $\issur_1>\issur_\bign$ for some $\rlo$-extremal $\slo$-perspective. First, \autoref{lem: generically distinct weights} tells us a dense set of $\rlo\in\Loss$ admit an $\rlo$-perspective which has associated weights $\reig=(\reig_i)_{i=1}^\bign$ that are all distinct. Second, by \autoref{lem: generically multiple R payoffs}, for any $\rlo$ of this form, a generic set of $\slo\in\Loss$ admit a unique $\slo$-perspective (which is then $\rlo$-extremal) such that $\simmat{\View}{\Viewb}\reig$ is not a constant vector. In particular, the set of pairs $(\slo,\rlo)\in\Loss^2$ each with distinct weights and yielding the desired $\issur_1>\issur_\bign$ property is dense in $\Loss^2$. Being semialgebraic by definition, it is generic by \autoref{lem: semialgebraic and dense implies generic} as desired.
\end{proof}

\begin{proof}[Proof of \autoref{lem: optimal investment switching}]
Let $\hsomemat:=\simmat{\hView}{\Viewb}$ and $\lsomemat:=\simmat{\lView}{\Viewb}$. Given \autoref{prop: eqm char} and \autoref{prop: eqlbm payoffs}, the R-best equilibrium value differs for S losses $\hslo$ and $\lslo$ for some capacity if and only if some $\litn\in[\bign-1]$ exists such that the sum of the $\litn$ highest entries of $\hsomemat\reig$ and $\lsomemat\reig$ are different. Because $\hsomemat$ and $\lsomemat$ being bistochastic means these two vectors have the same sum of entries ($\sum_{i=1}^\bign\reig_i$), this condition is equivalent to $\hsomemat\reig$ and $\lsomemat\reig$ not being rearrangements of one another. 
So we must show that generic $\reig\in\real^\bign_{++}$ have $\hsomemat\reig\neq\perm\lsomemat\reig$ for every permutation matrix $\perm\in\matnn$, and that all such $\reig$ admit the desired thresholds.

First, let us establish the genericity property. For any permutation matrix $\perm\in\matnn$, note that $\ker(\hsomemat-\perm\lsomemat)$ is a linear subspace of $\real^\bign$, and is a proper subspace because $\hView\nsim_\Viewb\lView$. But then, because $\real^\bign_{++}$ is a convex open subset of $\real^\bign$, it follows that $\real^\bign_{++}\setminus\ker(\hsomemat-\perm\lsomemat)$ is open and dense in $\real^\bign_{++}$ with Lebesgue-null complement in it---so, generic in it. Therefore, the intersection of this finite collection of sets (ranging over the finitely many permutation matrices) is also generic in $\real^\bign_{++}$. In what follows, fix some $\reig$ belonging to this generic set.

Let $\hissur$ and $\lissur$ be the rearrangements of $\hsomemat\reig$ and $\lsomemat\reig$, respectively, with $\hissur_1\geq\cdots\geq\hissur_\bign$ and $\lissur_1\geq\cdots\geq\lissur_\bign$. We have $\hissur\neq\lissur$ by definition of $\reig$, and $\hissur$ majorizes $\lissur$ because $\hView\succsim_\Viewb\lView$. Therefore, every $\litn\in[\bign-1]$ has $
\sum_{i=1}^{\litn}(\hissur_i-\lissur_i)\geq0$, but the inequality cannot hold with equality for every $\litn$. So let $\llitn$ and $\hlitn_0$ be the first and last index, respectively, in $[\bign-1]$ for which the inequality holds strictly, and let $\hlitn:=1+\hlitn_0\in[\bign]$. Now, let 
$(\hhcost,\hcost,\lcost,\llcost):=(\hissur_{\llitn},\lissur_{\llitn},\lissur_{\hlitn},\hissur_{\hlitn}).$ 
The definitions of $\llitn$ and $\hlitn$ tell us every $\litn\in[\bign]$ with $\litn<\llitn$ or $\litn>\hlitn$ has $\hissur_{\llitn}=\lissur_{\llitn}$, and hence 
$\hhcost\neq\hcost\geq\lcost\neq\llcost$. It then follows from the majorization inequalities that $\hhcost>\hcost\geq\lcost>\llcost$. We can now apply \autoref{lem: investment full solution} in each of these three regions. Costs in $(\llcost,\lcost)$ yield $\loptlitnb\geq\hlitn>\hoptlitnb$; costs in $(\hcost,\hhcost)$ yield $\loptlitnb<\llitn\leq\hoptlitnb$; and for costs outside of $[\llcost,\hhcost]$, optimal R-best-equilibrium investment is determined by the exact same first-order conditions for both $\hslo$ and $\lslo$.
\end{proof}

\section{Proofs for Discussion Section}

\begin{lemma}\label{lem: profiles with S best response}
Every $\slo$-projection of rank $\litn$ is equal to $\rstrat\sstrat$ for some R strategy $\rstrat$ and S best response $\sstrat$; and every such strategy profile yields such an $\slo$-projection.
\end{lemma}
\begin{proof}
The second assertion is proven in \autoref{lem: S best response}. Toward the first assertion, consider any $\slo$-projection $\strats$ with rank $\litn$. Because it has rank $\litn$, some $\rstrat\in\matnk$ exists with the same range---let the columns be any basis for said range. \autoref{lem: S best response} tells us some S best response $\sstrat$ exists, and this map is a rank-$\litn$ linear map such that $\rstrat\sstrat$ is an $\slo$-projection. But then, $\rstrat\sstrat=\strats$ as required, because an orthogonal projection on any inner-product space is determined by its range.
\end{proof}

\begin{proof}[Proof of \autoref{prop: persuasion}]
In what follows, call our main model \emph{the original game}, and let \emph{the common-interest game} refer to a modification of our main model in which R too has loss matrix $\slo$.

Below, we will establish that some S-preferred strategy profile exists; first, let us see the equivalence assuming existence. Indeed, in this case, a pure strategy profile is S-preferred if and only if it is a best pure-strategy Nash equilibrium of the common-interest game (in normal form). But both in the common-interest game and in the original game, a strategy profile is a Nash equilibrium of the normal form if and only if it is outcome equivalent to a Bayes Nash equilibrium. The equivalence then follows from \autoref{cor: S cares less about R}.

All that remains is existence of an S-preferred strategy profile. By \autoref{lem: S best response}, every R strategy admits an S-best response. So we can without loss show optimality only among strategy profiles in which S is best responding. Given \autoref{lem: profiles with S best response} and \autoref{lem: trace payoff calculation}, we need only see an optimum exists to the program 
$$\max_{\strats \text{ an $\slo$-projection of rank } \litn} \Tr\brac{\paren{2\slo\strats-\strats^\top\slo\strats}\mvar}.$$
Since the domain is compact and the objective continuous, some optimum exists.
\end{proof}

\begin{lemma}\label{lem: FOC for delegation}
Some optimum exists in the delegation problem, and every optimum has linear S strategy. If $(\sstrat,\rstrat)$ is an optimum, and $\strats=\rstrat\sstrat\in\matnn$, then  $$\strats^\top\brac{\rlo(I-\strats)\mvar\slo+\slo\mvar(I-\strats^\top)\rlo}(1-\strats)=0.$$ 
Hence, if $\strats$ is a $\inv{\mvar}$-projection, then it is also an $\rlo$-projection.
\end{lemma}
\begin{proof}
Any feasible profile, hence any optimum, in the delegation problem has linear S strategy by \autoref{lem: S best response}. 
Given \autoref{lem: profiles with S best response}, a committed R can induce $\act$ if and only if $\act=\strats\state$ for some $\slo$-projection $\strats$ of rank $\litn$. Her payoff from such a strategy profile is $\Tr\brac{\paren{2\rlo\strats-\strats^\top\rlo\strats}\mvar}$ by \autoref{lem: trace payoff calculation}.
She therefore solves the program 
$$\max_{\strats \text{ an $\slo$-projection of rank } \litn} \Tr\brac{\paren{2\rlo\strats-\strats^\top\rlo\strats}\mvar}.$$
Since the domain is compact and the objective continuous, some optimum exists.

Now, define $$(\stratsb,\mvarb,\rlob)=
(\rootslo\strats\inv{\rootslo},\ \rootslo\mvar\rootslo,\ \inv{\rootslo}\rlo\inv{\rootslo}).$$
The matrices $\mvarb$ and $\rlob$ are symmetric, and $\strats$ is $\slo$-orthogonal if and only if $\stratsb$ is an orthogonal matrix. Hence, 
\begin{eqnarray*}
\Tr\brac{\paren{2\rlo\strats-\strats^\top\rlo\strats}\mvar}
&=& \Tr\brac{\rootslo\mvar\paren{2\rlo\strats-\strats^\top\rlo\strats}\inv{\rootslo}} \\
&=& \Tr\paren{2\mvarb\rlob\stratsb-\mvarb\stratsb\rlob\stratsb} 
=\Tr\paren{2\mvarb\rlob\stratsb-\mvarb\stratsb\rlob\stratsb^\top} 
\end{eqnarray*}
Because any orthogonal matrix can be conjugated by $\rootslo$ to get an $\slo$-projection, any optimum for the delegation problem then maximizes $\Tr\paren{2\mvarb\rlob\stratsb-\mvarb\stratsb\rlob\stratsb^\top}$ over all rank-$\litn$ orthogonal projection matrices. Now, for any skew-symmetric matrix $\somesk\in\matnn$, the matrix exponential $e^{\somesk}$ is orthogonal, and so too is $e^{-\scalar\somesk}\stratsb e^{\scalar\somesk}$ for each $\scalar\in\real$. Optimality then requires $\Tr\paren{2\mvarb\rlob e^{-\scalar\somesk}\stratsb e^{\scalar\somesk}-\mvarb e^{\scalar\somesk}\stratsb e^{\scalar\somesk}\rlob e^{-\scalar\somesk}\stratsb e^{-\scalar\somesk}}$ to be maximized at $\scalar=0$. In particular, it requires the first-order condition 
\begin{equation}\label{eqn: delegation FOC}
\tfrac{\partial}{\partial\scalar}\big|_{\scalar=0}\Tr\paren{2\mvarb\rlob e^{-\scalar\somesk}\stratsb e^{\scalar\somesk}-\mvarb e^{-\scalar\somesk}\stratsb e^{\scalar\somesk}\rlob e^{-\scalar\somesk}\stratsb e^{\scalar\somesk}}=0
\end{equation}
to hold. To write \eqref{eqn: delegation FOC} more explicitly, note $e^{0\somesk}=I$ and $\tfrac{\partial}{\partial\scalar}\big|_{\scalar=0}e^{\scalar\somesk}=\somesk$, so that 
$$\tfrac{\partial}{\partial\scalar}\big|_{\scalar=0}\brac{e^{-\scalar\somesk}\stratsb e^{\scalar\somesk}}=I \stratsb \somesk + (-\somesk)\stratsb I = [\stratsb,\somesk].$$
Therefore, \begin{eqnarray*}
&&\tfrac{\partial}{\partial\scalar}\big|_{\scalar=0}\brac{2\mvarb\rlob e^{-\scalar\somesk}\stratsb e^{\scalar\somesk}-\mvarb e^{-\scalar\somesk}\stratsb e^{\scalar\somesk}\rlob e^{-\scalar\somesk}\stratsb e^{\scalar\somesk}}\\
&=& 
2\mvarb\rlob [\stratsb,\somesk]-\mvarb \curlyb{ (\stratsb)\rlob [\stratsb,\somesk]
+
[\stratsb,\somesk]\rlob (\stratsb)
}\\
&=& 
2\mvarb\rlob (\stratsb\somesk-\somesk\stratsb)-\mvarb \stratsb\rlob (\stratsb\somesk-\somesk\stratsb)
-\mvarb
(\stratsb\somesk-\somesk\stratsb)\rlob \stratsb,
\end{eqnarray*}
and hence skew symmetry of $\somesk$ implies \begin{eqnarray*}
&&\tfrac{\partial}{\partial\scalar}\big|_{\scalar=0}\Tr\paren{2\mvarb\rlob e^{-\scalar\somesk}\stratsb e^{\scalar\somesk}-\mvarb e^{-\scalar\somesk}\stratsb e^{\scalar\somesk}\rlob e^{-\scalar\somesk}\stratsb e^{\scalar\somesk}} \\
&=&\Tr\brac{2\mvarb\rlob (\stratsb\somesk-\somesk\stratsb)-\mvarb \stratsb\rlob (\stratsb\somesk-\somesk\stratsb)
-\mvarb
(\stratsb\somesk-\somesk\stratsb)\rlob \stratsb}\\
&=& \Tr\brac{\paren{2\mvarb\rlob \stratsb-2\stratsb\mvarb\rlob-\mvarb \stratsb\rlob \stratsb+\stratsb\mvarb \stratsb\rlob
-\rlob \stratsb\mvarb
\stratsb+\stratsb\rlob \stratsb\mvarb}\somesk}\\
&=& \Tr\brac{\paren{-2 \stratsb\rlob\mvarb-2\stratsb\mvarb\rlob
+\stratsb\rlob \stratsb\mvarb 
+\stratsb\mvarb \stratsb\rlob
+\stratsb\mvarb\stratsb\rlob 
+\stratsb\rlob \stratsb\mvarb}\somesk}\\
&=& -2\Tr\curlyb{\stratsb\brac{ \rlob(I-\stratsb)\mvarb 
+\mvarb(I-\stratsb)\rlob}\somesk}.
\end{eqnarray*}
Requiring \eqref{eqn: delegation FOC} to hold for every skew-symmetric matrix $\somesk$ then amounts to requiring that matrix $\stratsb\somemat$ be symmetric, where we define the symmetric matrix $$\somemat:=\rlob(I-\stratsb)\mvarb 
+\mvarb(I-\stratsb)\rlob.$$ 
But then symmetry of $\stratsb\somemat$ equivalently says $\somemat$ commutes with the orthogonal projection $\stratsb$, so that $\stratsb\somemat(I-\stratsb)=\somemat\stratsb(I-\stratsb)=0$. Hence, \begin{eqnarray*}
0 &=& \rootslo \stratsb\somemat(I-\stratsb) \rootslo \\
&=& \rootslo \stratsb\brac{\rlob(I-\stratsb)\mvarb 
+\mvarb(I-\stratsb)\rlob}(I-\stratsb) \rootslo \\
&=& \slo\strats\inv{\slo}\brac{\rlo(I-\strats)\mvar\slo 
+\slo\mvar(I-\slo\strats\inv{\slo})\rlo}(I-\strats) \\
&=& \strats^\top\brac{\rlo(I-\strats)\mvar\slo 
+\slo\mvar(I-\strats^\top)\rlo}(I-\strats),
\end{eqnarray*}
where the last equality follows from $\strats$ being an $\slo$-projection.

Finally, toward the last assertion, suppose $\strats$ is a $\inv{\mvar}$-projection in addition to being an $\slo$-projection, and is optimal in the delegation problem. In this case, 
$\strats\mvar\slo=\mvar\strats^\top\slo=\mvar\slo\strats,$ so that
\begin{eqnarray*}
0 &=&  \strats^\top\brac{\rlo(I-\strats)\mvar\slo 
+\slo\mvar(I-\strats^\top)\rlo}(I-\strats) \\
&=&  \strats^\top\rlo\mvar\slo(I-\strats) 
(I-\strats) 
+\strats^\top(I-\strats^\top)\slo\mvar\rlo(I-\strats) 
\\
&=&  \strats^\top\rlo\mvar\slo(I-\strats) 
+0\slo\mvar\rlo(I-\strats) 
 \hspace{1.2in} \text{ (since $\strats$ is idempotent)} \\
&=&  \strats^\top\rlo(I-\strats)\mvar\slo 
+0.
\end{eqnarray*}
But then $\mvar\slo$ being invertible means $\strats^\top\rlo(I-\strats)=0$. So $\strats$ is an $\rlo$-projection.
\end{proof}

\begin{proof}[Proof of \autoref{prop: delegation}]
We begin with some useful notation. Let $\Loss_1$ be the set of all $\slo\in\Loss$ such that an $\slo$-perspective is witnessed by $\bign$ distinct weights. Define the correspondence $\stratsset:\Loss\rightrightarrows\matnn$ by letting $\stratsset(\slo)$ be the set of rank-$\litn$ matrices that are simultaneously $\slo$-projections and $\inv{\mvar}$-projections. Define $\func:(\matnn)^2\to\matnn$ by letting $\func(\rlo,\strats):=\strats^\top\rlo\strats-\rlo\strats=(I-\strats)^\top\rlo\strats$. By \autoref{lem: S best response} and \autoref{lem: R best response}, every equilibrium (whatever $\rlo$ is) has a linear S strategy and yields $\act=\strats\state$ for some $\strats\in\stratsset(\slo)$; and \autoref{lem: FOC for delegation} tells us such an equilibrium can also be optimal in the delegation problem only if $\func(\rlo,\strats)$ zero. So the proposition will follow if we establish that generic $(\slo,\rlo)\in\Loss^2$ have $\func(\rlo,\strats)$ nonzero for every $\strats\in\stratsset(\slo)$.

Below, we will show that every $\slo\in\Loss_1$ and $\strats\in\stratsset(\slo)$ admit some $\rlo\in\Loss$ such that $\func(\rlo,\strats)\neq0$. Let us first see how the proposition would follow. Given $\slo\in\Loss_1$ and $\strats\in\stratsset(\slo)$, because $\func(\cdot,\strats)$ is linear, it would follow that a dense set of $\rlo\in\Loss$ have $\func(\rlo,\strats)\neq0$---indeed, any proper convex combination $\rlo$ of $\rlo^0,\rlo^1\in\Loss$ with $\func(\rlo^0,\strats)=0\neq\func(\rlo^1,\strats)$ has $\func(\rlo,\strats)\neq0$. But then (again given such $\slo$ and $\strats$), because the set of such $\rlo\in\Loss$ is semialgebraic, it is generic by \autoref{lem: semialgebraic and dense implies generic}. Now, for every $\slo\in\Loss_1$, it follows from  \autoref{lem: perspective representation of A matrices} that $\stratsset(\slo)$ is finite, and so (a finite intersection of generic sets being generic) generic $\rlo\in\Loss$ have $\{\func(\slo,\rlo,\strats)\}_{\strats\in\stratsset(\slo)}$ all nonzero. In particular, because $\Loss_1$ is dense in $\Loss$ by \autoref{lem: generically distinct weights}, a dense set of $(\slo,\rlo)\in\Loss^2$ has $\func(\rlo,\strats)$ nonzero for every $\strats\in\stratsset(\slo)$. Finally, this set is generic by \autoref{lem: semialgebraic and dense implies generic}. 

So now fix $\slo\in\Loss_1$ and $\strats\in\stratsset(\slo)$. It remains to see some $\rlo\in\Loss$ has $\func(\rlo,\strats)\neq0$. 
And indeed, because the idempotent matrix $\strats^\top$ has rank $\litn\in\{1,\ldots,\bign-1\}$, it has some eigenvectors $\view,\viewb\in\real^\bign$ with respective eigenvalues $1$ and $0$. So consider the loss matrix $\rlo:=\slo+(\view+\viewb)(\view+\viewb)^\top$. That $\strats$ is an $\slo$-perspective then tells us \begin{eqnarray*}
\func(\rlo,\strats)&=&\func(\slo,\strats)+\func\paren{(\view+\viewb)(\view+\viewb)^\top,\strats}\\
&=& 0 + \brac{(I-\strats^\top)(\view+\viewb)}\brac{\strats^\top(\view+\viewb)}^\top \\
&=& \viewb\view^\top,
\end{eqnarray*}
which is nonzero because $\view$ and $\viewb$ are.
\end{proof}

\end{document}